\documentclass{amsart}

\usepackage{amsmath, amssymb, amsthm}
\usepackage[mathcal]{eucal}
\usepackage{enumerate}
\usepackage{graphicx}
\usepackage{verbatim}
\usepackage[dvipsinames]{xcolor}
\definecolor{midnightblue}{HTML}{0059b3}
\definecolor{chromered}{HTML}{f14233}
 \RequirePackage[colorlinks,citecolor=midnightblue,urlcolor=midnightblue,breaklinks=true,
linkbordercolor=midnightblue,linkcolor=midnightblue]{hyperref}

\newtheorem{theorem}{Theorem}
\newtheorem{lemma}{Lemma}

\def\Z{\ensuremath{\mathbb Z}}
\def\C{\ensuremath{\mathbb C}}
\def\N{\ensuremath{\mathbb N}}
\def\R{\ensuremath{\mathbb R}}
\def\T{\ensuremath{\mathbb T}}
\def\D{\ensuremath{\mathbb D}}

\renewcommand{\tilde}{\widetilde}
\renewcommand{\epsilon}{\varepsilon}
\renewcommand{\phi}{\varphi}
\renewcommand{\bar}{\overline}
\newcommand{\sgn}{\operatorname{sign}}
\newcommand{\bH}{\mathbf{H}}
\newcommand{\Id}{\operatorname{Id}}
\newcommand{\esssup}{\operatorname{ess} \,\operatorname{sup}}
\renewcommand{\equiv}{:=}
\newcommand{\dist}{\operatorname{dist}}
\renewcommand{\Im}{\operatorname{Im}}
\renewcommand{\Re}{\operatorname{Re}}
\newcommand{\cst}{\eta}

\numberwithin{equation}{section}

\newcommand\numberthis{\addtocounter{equation}{1}\tag{\theequation}}

\title{Quantum signal processing and nonlinear Fourier analysis}

\author{Michel Alexis}
\author{Gevorg Mnatsakanyan}
\author{Christoph Thiele}

\address{Mathematical Institute, 
	University of Bonn,
	Endenicher Allee 60, 53115 Bonn,
	Germany}
\email{alexis@math.uni-bonn.de}
	 \email{gevorg@math.uni-bonn.de}
  \email{thiele@math.uni-bonn.de}


\begin{document}

\begin{abstract}
Elucidating a connection with nonlinear Fourier analysis (NLFA), we extend a well known algorithm in quantum signal processing (QSP) to represent measurable signals by square summable sequences.
Each coefficient of the sequence is Lipschitz continuous as a function of the signal.

\end{abstract}

\maketitle

\section{Introduction}

A signal in this paper is a function from the interval $I=[0,1]$ to 
the interval $(-2^{-\frac 12}, 2^{-\frac 12})$.
In quantum signal processing, one represents such a
signal as the imaginary part of one entry of an ordered
product of unitary matrices. The factors of this product alternate between matrices depending on the functional parameter $x\in I$
and matrices depending on a sequence $\Psi$  of scalar parameters $\psi_n$ which 
are tuned so that the product represents a given signal.
We are interested in the particular representation of this type
proposed by \cite{low2017} and extended to infinite 
absolutely summable sequences in \cite{Linlin}.
Our main observation is that after some change of variables, the map sending the sequence $\Psi$ to the signal is identified as the nonlinear Fourier series described in \cite{tsai}.
Indeed, this nonlinear Fourier series as well as variants 
including one  with $SU(1,1)$ 
matrices in \cite{TaoThiele2012} have been studied 
for a long time in different contexts such as orthogonal polynomials \cite{simon}, Krein systems \cite{Denisov}, scattering transforms
\cite{bealscoifman, sylvesterwinebrenner}
or AKNS systems \cite{akns}.

In particular, transferring knowledge from
nonlinear Fourier analysis, we extend
the theory in \cite{Linlin}
from absolutely summable to square summable $\Psi$ using a nonlinear version of the Plancherel identity. We obtain a representation of  measurable signals
by square summable sequences $\Psi$.
This representation extremizes a certain inequality of Plancherel type.

To state our main result, Theorem \ref{main}, we make some formal definitions.
Given $\epsilon>0$, define the signal space $\mathbf{S}_{\epsilon}$ to be the set of real valued measurable functions $f$ on $[0,1]$ 
that satisfy the bound  
\begin{equation}\label{upperboundb}
    \sup_{x\in [0,1]}|f(x)|\le 2^{- \frac{1}{2}}-\epsilon .
\end{equation}
We equip $\mathbf{S}_{\epsilon}$ with the metric induced by the Hilbert space norm 
\begin{equation}\label{hsnorm}
\left \| f \right \| \equiv \left ( \frac{2}{\pi} \int\limits_{0} ^1 \left | f(x) \right |^2 \frac{dx}{\sqrt{1-x^2}} \right )^{\frac{1}{2}} \, .
\end{equation}

Let $\mathbf{P}$ be the space of sequences
$\Psi=(\psi_k)_{k\in \N}$ of numbers $\psi_k\in (-\frac \pi 2,\frac \pi 2)$. We equip $\mathbf{P}$ with the metric induced by the $L^{\infty}$-norm
\[\|\Psi\|_\infty = \sup_{k\in \N}|\psi_k|.\]

For  $x\in [0,1]$, define
\begin{equation}\label{eq:def_Wonly}
       W(x) := \begin{pmatrix}
        x & i\sqrt{1-x^2} \\
        i\sqrt{1-x^2} & x
    \end{pmatrix} ,\quad
    Z = \begin{pmatrix}
        1 & 0 \\ 0 & -1
    \end{pmatrix} 
    .
\end{equation}
For $\Psi \in \mathbf{P}$ and $x\in [0,1]$, define recursively
\begin{equation}\label{u0}
U_0(\Psi,x)=e^{i\psi_0 Z}
\end{equation}
and
\begin{equation}\label{ud}
U_d(\Psi,x)=
e^{i\psi_{d} Z}W(x) U_{d-1}(\Psi,x) W(x)e^{i\psi_{d} Z}.
\end{equation}
Define $u_d(\Psi,x)$ to be the upper left entry of $U_d(\Psi,x)$.

\begin{theorem}\label{main}
 Let $\epsilon>0$. For each $f \in \mathbf{S}_{\epsilon}$, there exists a unique 
sequence $\Psi \in \mathbf{P}$ such that 
\begin{equation}
\label{plancherel}
\sum_{k\in \Z}\log(1+\tan^2\psi_{|k|})
=
-\frac{2}{\pi}\int_{0}^1 \log|1-f(x)^2| \frac{dx}{\sqrt{1-x^2}} 
\end{equation}
and 
$\Im(u_d(\Psi,x))$
converges with respect to the norm \eqref{hsnorm} to the function $f$
as $d$ tends to $\infty$. For two functions $f,\tilde{f} \in \mathbf{S}_\epsilon$ with corresponding sequences $\Psi, \tilde{\Psi}$ as above, we have the Lipschitz bound
\begin{equation}\label{lipschitzbound}
\|\Psi-\tilde{\Psi}\|_\infty \le 7.3 \epsilon^{-\frac{3}{2}}\|f-\tilde{f}\|.
\end{equation}
\end{theorem}
\begin{figure}
    \centering
\includegraphics[width=0.6\textwidth]{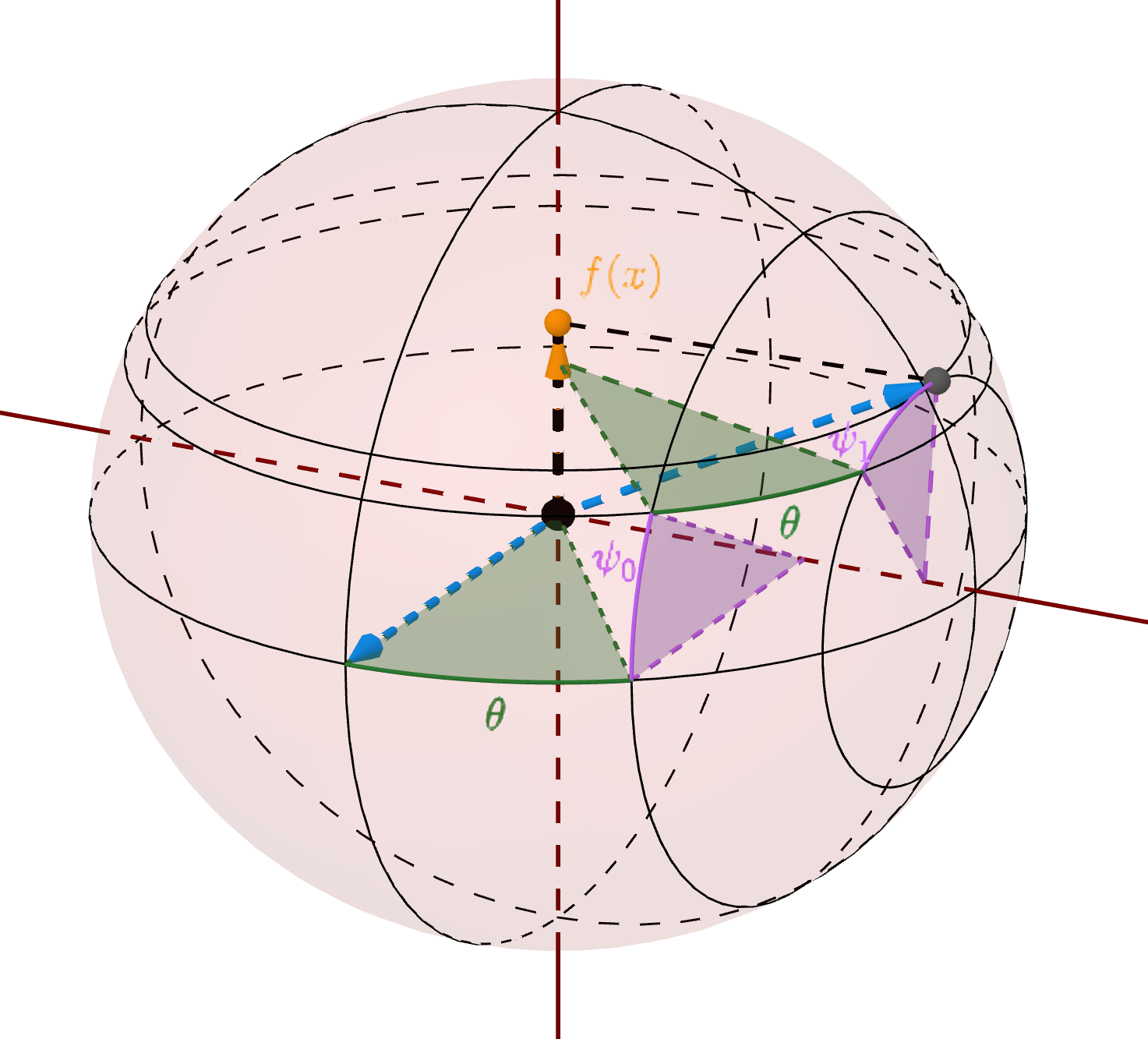}
    \caption{Illustration of QSP}
\label{fig:sphere}
\end{figure}

Figure \ref{fig:sphere} is a simplified cartoon of QSP, conflating for illustrative purpose the group $SO(3)$ with its double-cover  $SU(2)$ and ignoring for simplicity the reflection symmetry in the product \eqref{ud}. For a given signal $f$, Theorem \ref{main} provides tuning parameters $\psi_j$ with which we can then evaluate $f$ at $x = \cos \theta$ as follows. We alternatingly rotate the horizontal blue vector by $\theta$ about the vertical axis, an action generated by the Pauli matrix $X$ defined in Section
\ref{section:QSP_to_NLFT_finite}, and by the consecutive tuning parameters $\psi_j$ about the horizontal brown axis, an action generated by the Pauli matrix $Z$. The resulting rotated blue vector has height $f(x)$.

Our proof provides an algorithm to compute $\psi_k$ via a Banach fixed point iteration that converges exponentially fast with rate depending on $\epsilon$. The iteration step requires the application of a Cauchy projection, which in practice may be computed using a fast Fourier transform. 

The weight $(1-x^2)^{-\frac 12}$ in \eqref{plancherel} has a singularity at one but not at zero. This asymmetry 
arises because our theory works naturally with $f$ extended to an even function on $[-1,1]$. 

After developing the relevant parts
of nonlinear Fourier analysis, we prove  Theorem \ref{main} in Section \ref{mainproof}.
A relaxation of the threshold \eqref{upperboundb} will be discussed in a forthcoming paper.

The literature both on QSP and NLFA is extensive and we do not try to give a complete overview here. 
Our first reference to the QSP algorithm discussed here is \cite{low2017}, which was interested in an optimal algorithm for Hamiltonian simulation. Various interesting properties of QSP are discussed in \cite{symmetricphasefactors, childs2018toward}. \cite{su11qsp} introduces an $SU(1,1)$ variant of QSP. For the task of computing the potential $(\psi_n)$ for a given target function $f$, several algorithms have been proposed including the so-called factorization method \cite{factorizationmethod1, factorizationmethod2, FindingAnglesforQSP, StableFactorization}, an optimization algorithm \cite{EnergylandscapeofQSP}, fixed point iteration \cite{Linlin} and Newton's method \cite{linlinrecentnumerics}. The factorization method in the context of nonlinear Fourier series is called the layer-stripping formula and is discussed below. The papers \cite{Linlin} and \cite{EnergylandscapeofQSP} develop the $\ell^1$ and $\ell^2$ theories for QSP with many interesting theoretical results. Many of these results are
implicit in 
our discussion of NLFA in the present paper.

Discrete NLFA for the $SU(1,1)$ model was studied in \cite{TaoThiele2012} 
with particular emphasis on transferring analytic estimates 
for the linear Fourier transform 
to the nonlinear setting.
For the $SU(2)$ model, a similar
discussion appears in \cite{tsai}. Some important contributions to the 
quest for analogs of classical linear inequalities were made by \cite{ChristKiselev, Oberlin2012, diogo, KovacDiogo, MuscaluTaoThiele}, namely providing maximal and variational Hausdorff-Young inequalities and discussing Carleson-type theorems for the $SU(1,1)$ model of the nonlinear Fourier transform and some variants. For a discussion of some recent results and open questions see \cite{diogosurvey}.

The interest of the third author and subsequently
the other authors in quantum signal processing was initiated during an inspiring talk by
L. Lin at a delightful conference at ICERM on Modern Applied and Computational Analysis. In particular, we dedicate this result to R. Coifman, who anticipated at the conference that QSP is some sort of nonlinear Fourier analysis. The third author acknowledges an invitation to the Santal\'o Lecture 2022, where he gave an introduction to nonlinear Fourier analysis. 
The authors acknowledge support by the Deutsche Forschungsgemeinschaft (DFG, German Research Foundation) under Germany's Excellence Strategy -- EXC-2047/1 -- 390685813 as well as CRC 1060. We also thank Jiasu Wang for pointing out typos in the first ArXiv posting of this article.

\section{The nonlinear Fourier transform}
We are mainly interested in 
nonlinear Fourier series. However, we start with an excursion to the nonlinear Fourier transform on the real line, which is a multiplicative and non-commutative version of the linear Fourier transform.

Recall the linear Fourier transform 
$$\widehat{f}(\xi):=\int_\R f(x) e^{-2\pi i x\xi} dx.$$
This integral is understood to be a Lebesgue integral if
$f$ is in $L^1(\R)$. If $\widehat{f}$ is also in $L^1(\R)$, then
both $f$ and $\widehat{f}$ can be seen to be
in $L^2(\R)$ and one has the Plancherel identity
$$\|\widehat{f}\|_2=\|{f}\|_2.$$
The Plancherel identity holds for $f$ in a dense subset of $L^2(\R)$, and one can use it to extend the Fourier transform 
to a unitary map from $L^2(\R)$ to itself.
This definition in $L^2(\R)$ coincides with the integral definition  when
$f$ is in $L^2(\R)\cap L^1(\R)$.

The integral in the definition of the Fourier transform is an additive process over continuous time $x$.
This process can alternatively be expressed by a differential evolution equation for the partial Fourier integrals
$$S(\xi,x)=\int_{-\infty}^x f(t) e^{-2\pi i t\xi}\, dt, $$
namely
$$\partial_x S(\xi,x)= f(x) e^{-2\pi i x\xi}$$
with the initial condition
$$S(\xi,-\infty)=0$$
and the final state 
$$S(\xi,\infty)=\widehat{f}(\xi).$$
If $f\in L^1(\R)$, the required analytic facts such as solvability of the differential equation and limits as $x$ tends to $\pm \infty$ can be elaborated with standard methods.

Exponentiation turns this additive process into a multiplicative process. Define
$$G(\xi,x)=e^{S(\xi,x)} .$$
Then $G$ satisfies the differential equation
\begin{equation}\label{nlftc}
\partial_x G(\xi,x)= G(\xi,x) f(x) e^{-2\pi i x\xi}
\end{equation}
with the initial condition
$$G(\xi,-\infty)=1$$
and the final state
$$G(\xi,\infty)=e^{\widehat{f}(\xi)}.$$

In the above scalar valued setting, the multiplicative perspective is of an  artificial nature.
However, the multiplicative process allows 
for matrix valued generalizations, which lead to substantially 
different nonlinear Fourier transforms.
For these generalizations, the complex factor $f(x) e^{-2\pi i x\xi}$ in $\C$
in \eqref{nlftc} needs to be replaced by a matrix factor. The most basic choices of such matrix factors come from real linear embeddings of $\C$ into three dimensional Lie algebras, 
in particular the ones associated with the Lie groups $SU(1,1)$ and $SU(2)$. 

The most common $SU(1,1)$ model of the nonlinear Fourier transform is described by the differential equation
\begin{equation}\label{nlftsu11}
\partial_x G(\xi,x)= G(\xi,x)\left(\begin{array}{cc} 0 & f(x) e^{-2\pi i x\xi}
\\ \overline{f(x) e^{-2\pi i x\xi}} & 0 \end{array}\right)
\end{equation}
with the initial condition
$$G(\xi,-\infty)=\left(\begin{array}{cc} 1 & 0
\\ 0 & 1 \end{array}\right)$$
and the final state defined to be the $SU(1,1)$ nonlinear Fourier transform  of $f$,
\begin{equation}\label{su11nlft}
G(\xi,\infty)=\left(\begin{array}{cc} a(\xi) & b(\xi)
\\ \overline{b(\xi)} & \overline{a(\xi)} \end{array}\right).\end{equation}
As the matrix factor in \eqref{nlftsu11} is in the Lie Algebra of $SU(1,1)$,  the solution to the differential equation stays in $SU(1,1)$. This explains the particular structure of the matrix in \eqref{su11nlft} and we also have 
$$|a(\xi)|^2-|b(\xi)|^2=1.$$

Analogous to the linear situation, solvability of the differential equation with limits as $x$ tends to $\pm \infty$ is elementary
for $f\in L^1(\R)$.
By Picard iteration, a solution can be written as
the limit of recursively defined approximations $G_{k}$ with
$$G_{0}(\xi,x)=\left(\begin{array}{cc} 1 & 0
\\ 0 & 1 \end{array}\right)$$
and for $k>0$
$$G_{k}(\xi,x)=\left(\begin{array}{cc} 1 & 0
\\ 0 & 1 \end{array}\right)+\int_{-\infty}^x 
G_{k-1}(\xi,t_k)\left(\begin{array}{cc} 0 & f( t_k) e^{-2\pi i t_k\xi}
\\ \overline{f(t_k) e^{-2\pi i t_k\xi}} & 0 \end{array}\right)
\, dt_k.$$
In particular, $G_{k}-G_{k-1}$ is $k$-linear in $f$.
If $k$ is even, the $k$-linear term is diagonal with 
upper left entry
\begin{equation}\label{diagftentry}
\int_{-\infty<t_1<t_2<\dots <t_k<\infty}
\prod_{j=1}^{k/2} \overline{f(t_{2j})}{f(t_{2j-1})}e^{2\pi i \xi (t_{2j}-t_{2j-1})}\, dt_{2j-1}dt_{2j}
\end{equation}
and lower right entry the complex conjugate of \eqref{diagftentry}.
If $k$ is odd, then the $k$-linear term is anti-diagonal 
with upper right entry
\begin{equation}\label{ftoffdiag}\int_{-\infty<t_1<t_2<\dots <t_k<\infty}
f(t_{k})e^{-2\pi i \xi t_k} \prod_{j=1}^{(k-1)/2} 
\overline{f(t_{2j})} f(t_{2j-1})e^{2\pi i \xi (t_{2j}-t_{2j-1})}\, dt_{2j-1}dt_{2j}
\end{equation}
and lower left entry the complex conjugate of \eqref{ftoffdiag}. Note that \eqref{diagftentry} and \eqref{ftoffdiag} are the terms involved in the multilinear expansions of $a$ and $b$, which have first order approximation of  the constant function $1$ and the linear Fourier transform 
of $f$, respectively.

The entries  \eqref{diagftentry} and \eqref{ftoffdiag} are bounded in absolute value by the integrals
\begin{equation}\label{factorial}
\int_{-\infty<t_1<t_2<\dots <t_k<\infty}
\prod_{j=1}^{k} |f(t_{j})| dt_j =\frac 1{k!} \|f\|_1^k.
\end{equation}
Hence the nonlinear Fourier transform is a real analytic map 
from $L^1(\R)$ to the space $L^\infty(\R, \C^{2\times 2})$.
Moreover, the inverse linear Fourier transform of
\eqref{diagftentry} can be written
with the Dirac $\delta$ and the functional variable $x$ as
\begin{equation}\label{ftdiagftentry}
\int_{-\infty<t_1<t_2<\dots <t_k<\infty}
\delta(x+\sum_{j=1}^{k/2} t_{2j}-t_{2j-1})
\prod_{j=1}^{k/2} \bar{f(t_{2j})}f(t_{2j-1}) \, dt_{2j-1}dt_{2j},
\end{equation}
and similarly for \eqref{ftoffdiag}.
The function 
\eqref{ftdiagftentry} is again in $L^1(\R)$ with norm bounded as in
\eqref{factorial}.
Hence the nonlinear Fourier transform is a real analytic map 
from $L^1(\R)$ to $A(\R, \C^{2\times 2})$, the matrix valued functions with entries in the Wiener space $A(\R)$, which is the linear Fourier transform of ${L^1}(\R)$.

With more work, one can also show that the nonlinear Fourier transform extends to
an analytic map from $L^p(\R)$ into a suitable space \cite{ChristKiselev, Oberlin2012} for $1<p<2$. At $p=2$,
the $SU(1,1)$ nonlinear Fourier transform can be defined 
by a similar density argument as in the linear case using the
nonlinear analogue of the Plancherel identity
\begin{equation}\label{eq:int_Planch}\|f\|_2^2=2 \int \log{|a(\xi)|}\, d\xi=\int_\R \log(1+|b(\xi)|^2) \, d\xi ,
\end{equation}
which we will elaborate on below after \eqref{diagftentry2}. 
However, unlike the linear setting, one obtains neither an injective map
on $L^2(\R)$, nor a real analytic map on $L^2(\R)$ in any suitable sense. See \cite{TaoThiele2012} in the discrete setting and \cite{MTT} for references on these respective phenomena.

The $SU(2)$ model of the nonlinear Fourier transform is described by the solution to the differential equation
\begin{equation}\label{nlftsu2}
\partial_x G(\xi,x)= G(\xi,x)\left(\begin{array}{cc} 0 & f(x) e^{-2\pi i x\xi}
\\ -\overline{f(x) e^{-2\pi i x\xi}} & 0 \end{array}\right)
\end{equation}
with the initial condition
$$G(\xi,-\infty)=\left(\begin{array}{cc} 1 & 0
\\ 0 & 1 \end{array}\right)$$
and whose final state is the $SU(2)$ nonlinear Fourier transform  of $f$
\begin{equation}\label{su2nlft}
G(\xi,\infty)=\left(\begin{array}{cc} a(\xi) & b(\xi)
\\ -\overline{b(\xi)} & \overline{a(\xi)} \end{array}\right).
\end{equation}
Here, the matrix in \eqref{su2nlft}  is in $SU(2)$
for each $\xi\in \R$, and in particular 
$$|a(\xi)|^2+|b(\xi)|^2=1.$$
The $L^p$ theory for $p<2$ in so far as discussed above is largely analogous to the case of $SU(1,1)$ but with suitable changes of
signs in the multi-linear terms.
The analogue of Plancherel however is the weaker information 
\begin{equation}\label{limita}
    \|f\|_2^2=  \lim_{\xi \to  i\infty}  2\pi i\xi \log{(a(\xi))} ,
\end{equation} 
where $\xi$ tends to $\infty$ along the imaginary axis in the upper half plane, or more generally through any ray from the origin strictly in the upper half plane. This can be shown by doing an asymptotic expansion
\[2\pi i\xi  \log(a(\xi))=  c +O(|\xi|^{-1})\]
along such a ray as in \cite{MuscaluTaoThiele} and observing  that it is only the bilinear term in the multilinear expansion of $a$ that contributes to $c$. The bilinear term of $a$, now the negative of the bilinear term of the $SU(1,1)$ case, is equal to
\begin{equation}\label{diagftentry1}
-\int_{-\infty<t_1<t_2 < \infty}
 \overline{f(t_{2})}{f(t_{1})}e^{2\pi i \xi (t_{2}-t_{1})}\, dt_{1}dt_{2}
=-\int_{s>0} \int_t
 \overline{f(t+s)}{f(t)}e^{2\pi i \xi s}\, dtds.
\end{equation}
Multiplying by $2\pi i \xi$ and using that 
\begin{equation}\label{appunit}-2\pi  i\xi  e^{2\pi i\xi s}1_{\{s>0\}}
\end{equation}
is an approximating unit converging to 
the Dirac delta  as $\xi$ tends to infinity along a ray in the upper half plane, we obtain 
\begin{equation}\label{diagftentry2}\lim _{\xi \to i\infty}
-2\pi i\xi  \int_{s>0} \int_t
 \overline{f(t+s)}{f(t)}e^{ 2\pi i\xi  s}\, dtds =\int \overline{f(t)} f(t)\, dt.
\end{equation}
This shows \eqref{limita}.

In the $SU(1,1)$ setting, where $\log(a)$ 
has an analytic extension to the upper half plane, one can 
use a contour integral over a large semicircle in the upper half plane to express the analogue of the limit \eqref{limita} by an
integral as in 
\eqref{eq:int_Planch}.  Here, in the $SU(2)$ case,  $\log(a)$ is 
in general not
analytic in the upper half plane due to zeros of $a$ and one cannot as easily express the limit by an integral. Instead, one resorts to tools such as
factorization into inner and outer functions \cite{garnett}.

\section{Nonlinear Fourier series}

Passing from functions on $\R$
to sequences $F$ on $\Z$, the Fourier transform, which we now call Fourier series, 
no longer lives on $\R$ but on the unit circle $\mathbb{T}:=\{z\in \C: |z|=1\}$. 
We slightly  misuse the notion of Fourier series here,
usually this notion is reserved  for the inverse of the map that we call Fourier series here.
 
There are nonlinear Fourier series with values in $SU(1,1)$, this is discussed in \cite{TaoThiele2012}, and  nonlinear Fourier series with values in $SU(2)$ discussed
in \cite{tsai}.
We focus here on the SU(2) model, which is relevant to the QSP model  in Theorem \ref{main}.

The linear Fourier series of a sequence $F=(F_n)_{n\in \Z}$ with finite support is defined as
$$\widehat{F}(z)=\sum_{n\in \Z} F_nz^n .$$
The analogy with the Fourier transform becomes apparent when writing $z=e^{-2\pi i \xi}$ for some 
$\xi\in \R$. Indeed, if we define a measure $f$ on the real line as 
$$f(x)=\sum_{n\in \Z} F_n\delta(x-n),$$
then
$$\widehat{f}(\xi)=\int_\R \sum_{n\in \Z} F_n\delta(x-n)e^{-2\pi i \xi x}\, dx
$$
\[=\int_\R \sum_{n\in \Z} F_n\delta(x-n)e^{-2\pi i \xi n}\, dx= \sum_{n\in \Z} F_ne^{-2\pi i \xi n}=\widehat{F}(z).\]
The nonlinear analog becomes an ordered product of matrices
described below.
We will be interested in meromorphic extensions beyond the 
circle $\T$, hence we consider the Riemann sphere $\C\cup \{\infty\}$
where $\infty$ is the reciprocal of $0$. For a subset 
$\Omega$ of the Riemann sphere we define the reflected set
\begin{equation}
    \Omega^*=\{\overline{z^{-1}}:z\in \Omega\}.
\end{equation}
For a function $a$ on $\Omega$ we define $a^*$ on $\Omega^*$ 
by
\begin{equation}
    a^*(z)=\overline{a(\overline {z^{-1}})}.
\end{equation}
We note that $(\Omega^*)^*=\Omega$ and $(a^*)^*=a$.
If $z\in \mathbb{T}$, then $a^*(z)=\overline{a(z)}$. 
Define the open unit disc
\[\mathbb{D} \equiv \left \{ z \in \C ~:~ \left |z \right | < 1 \right \} ,
\]
The function $a$ is analytic on $\D^*$ precisely if $a^*$ is analytic on $\D$. We have  
\begin{equation}\label{ainftyastar0}
a(\infty)=\overline{a^*(0)}.
\end{equation}
If $a$ is analytic on $\D^*$ and continuous up to the boundary
$\T$ of $\D^*$, then we have the mean value theorem 
\begin{equation}\label{inftymeanvalue}
a(\infty)=\overline{a^*(0)}=\overline{\int_\T a^*}=\int_\T a, 
\end{equation}
where we denote by
$$\int_\T a =  \int_0^{1} a(e^{2\pi i \theta})\, d\theta$$
the mean value of $a$ on $\T$, i.e., the constant term in the Fourier expansion of $a$.

For a sequence $F:\Z\to \C$ with finite support, define
the meromorphic matrix valued function $G$ on the Riemann
sphere by the recursive equation
\begin{equation}\label{nlfsrec}
G_k(z)= G_{k-1}(z)\frac{1}{\sqrt{1+|F_k|^2}}\left(\begin{array}{cc} 1 & F_k z^{k}
\\ -\overline{F_k} z^{-k} & 1 \end{array}\right)
\end{equation}
with the initial condition
\[\lim_{k\to -\infty} G_k(z)= \left(\begin{array}{cc} 1 & 0
\\ 0 & 1 \end{array}\right) , \] 
and define the $SU(2)$ nonlinear Fourier series 
\begin{equation}\label{nlfsdef}
G(z)=\lim_{k \to \infty} G_k(z)= \left(\begin{array}{cc} a(z) & b(z)
\\ -{b^*(z)} & {a^*(z)} \end{array}\right)  .
\end{equation}
Existence of the limit as $k\to \pm \infty$ is trivial
thanks to the finite support of $F$, which makes the sequence 
$G_k(z)$ eventually constant in $k$. The matrix factors in
\eqref{nlfsrec} are in $SU(2)$
on $\T$ and hence so is their product. In particular, 
\[a(z)a^*(z)+b(z)b^*(z)=1\]
on $\T$ and as well on the Riemann sphere by analytic continuation. 

Under the analogous formal transformation as above, the nonlinear Fourier series becomes  
the $SU(2)$ nonlinear Fourier transform for the measure
$$f(x)=\sum_{n\in \Z} f_n\delta(x-n),$$
where $f_n=\arctan(|F_n|)F_n |F_n|^{-1}$. This value of $f_n$ 
arises from the model computation
\[\exp\left(\begin{array}{cc} 0 & f_0  
\\ -\overline{f_0} & 0\end{array}\right)=
\left(\begin{array}{cc} \cos |f_0| &  {f_0}{|f_0|^{-1}}\sin |f_0|
\\ -{\overline{f_0}}{|f_0|^{-1}}\sin |f_0| & \cos |f_0|  \end{array}\right)
\]
\[=\cos|f_0| \left(\begin{array}{cc} 1 &  {f_0}{|f_0|^{-1}}\tan |f_0|
\\ -{\overline{f_0}}{|f_0|^{-1}}\tan |f_0| & 1  \end{array}\right)
=\frac 1{1+|F_0|^2}\left(\begin{array}{cc} 1 & F_0
\\ -\overline{F_0}  & 1 \end{array}\right).
\]

We write the $SU(2)$ nonlinear Fourier series of the sequence $F$ on $\Z$ as
\[{\overbrace {F}}:=(a,b)\]
with $a$ and $b$ as defined in \eqref{nlfsdef}. We identify the row vector $(a,b)$ with the matrix function as in \eqref{nlfsdef}. In particular, we write the product
\[(a,b)(c,d)=(ac-bd^*, ad+bc^*).\]

We  describe some properties of the nonlinear $SU(2)$ Fourier
series, following  \cite{tsai}
and the analogous arguments in \cite{TaoThiele2012}.
The first theorem describes some basic transformation properties analogous to transformation properties of the linear Fourier series. To better understand the analogy, recall from the analogous discussion of the nonlinear Fourier transform that the first order approximations of $a$ and $b$ are one and the linear Fourier series, respectively.

\begin{theorem}\label{SU2basictransformations}
Let $F, H$ be complex valued finitely supported sequences on $\Z$ and let
\[\overbrace{F}(z)=(a,b).\]
If all entries of $F$ except possibly the zeroth entry vanish, then 
\begin{equation}
    (a(z),b(z))=(1+|F_0|^2)^{-\frac 12}(1, F_0).
\end{equation}
If $H_n=F_{n-1}$, then
\begin{equation}\label{eq:potentialshift}
    \overbrace{H}(z)=(a(z),zb(z)).
\end{equation}
If the support of $F$ is entirely to the left of the support of $H$, then
\begin{equation}\label{leftright}
\overbrace{F+H}=\overbrace{F} \overbrace{H}.
\end{equation}
If $|c|=1$, then
\begin{equation}\label{eq:potentialmodulate}
    \overbrace {cF}=(a,cb).
\end{equation}
If $H_n=F_{-n}$, then
\begin{equation}\label{eq:potentialreflection}
    \overbrace{H}(z)=(a^*(z^{-1}), b(z^{-1})).
\end{equation}
If $H_n=\overline{F_n}$, then
\begin{equation}\label{eq:potentialconjugate}
    \overbrace{H}(z)=(a^*(z^{-1}), b^*(z^{-1})).
\end{equation}
\end{theorem}

Note that the
properties in Theorem \ref{SU2basictransformations} are sufficient to uniquely determine
the map from $F$ to $\overbrace{F}$.
The next theorem describes 
the range of this map on the space  of sequences with finite support.
Let $l(M,N)$ be the space of all complex valued sequences $F$
on $\Z$ which are supported on the interval $M\le k\le N$
in the strict sense that $F(M)\neq 0$ and $F(N)\neq 0$.

\begin{theorem}\label{finitethm}

Let $M\le N$. The $SU(2)$ nonlinear Fourier series maps
$l(M,N)$ bijectively to the space of pairs $(a,b)$  such that
$b$ is the linear Fourier series of a sequence in $l(M,N)$ and
$a$ is the linear Fourier series of a sequence in $l(M-N,0)$
with $0<a(\infty)$ and 
\begin{equation}\label{aastarbbstar}
     aa^*+bb^*=1.
\end{equation}
Moreover, we have the identity
\begin{equation}\label{finiteplancherel}
a(\infty)=\prod_{n\in \Z} (1+|F_n|^2)^{-1/2} .
\end{equation}
\end{theorem}
Note that \eqref{aastarbbstar} implies that $a$ and $b$ have no common zeros in the Riemann sphere. Moreover,
$|a|$ and $|b|$ are bounded by $1$ on $\T$ and $a(\infty)\le 1$
with equality only if $b=0$ and $F=0$.

Note that if $a$ does not have zeros in $\mathbb{D}^*$, then
$\log(a)$ is analytic in $\mathbb{D}^*$ and the real part of
\eqref{inftymeanvalue} gives
\begin{equation}\label{acontour}
    \log |a(\infty)| = \int_\T \log |a| .
\end{equation}
Multiplying by $-2$ and using \eqref{finiteplancherel} and \eqref{aastarbbstar}, we obtain
\begin{equation}\label{outerplancherel}
    \sum_{n\in \Z} \log (1+|F_n|^2) = -\int_\T \log (1- |b|^2) 
\end{equation}
in analogy to \eqref{eq:int_Planch}. If $a$ has zeros in $\mathbb{D}^*$,
then we have only the inequality
\begin{equation}
    \sum_{n\in \Z} \log (1+|F_n|^2) \ge  -\int_\T \log (1- |b|^2),
\end{equation}
which can be obtained by applying the mean value theorem
\eqref{inftymeanvalue}  to the logarithm of the quotient of
$a^*$ divided by the Blaschke product of its zeros \cite{garnett}.

As $\log(1+x)$ is comparable to $x$ for small $x$, under suitable pointwise smallness assumptions on $F$ and $b$ 
and absence of zeros of $a$ in $\D^*$ we obtain from \eqref{outerplancherel} that
$\|F\|_{l^2(\Z)}$ and $\|b\|_{L^2(\T)}$ are comparable, in analogy to the linear situation.


\section{Quantum signal processing for finite sequences}\label{section:QSP_to_NLFT_finite}

In this section, we relate at the level of finite sequences the nonlinear Fourier series 
to QSP.

Let $\Psi$ be in $\mathbf{P}$ as in Theorem \ref{main}.
Let $F_n$ for $n\in \Z$ be  defined by 
\begin{equation}
    F_n=i\tan(\psi_{|n|})
\end{equation}
and note that $(F_n)$ is even and purely imaginary, that is, for all $n\in \Z$,
\[
F_{-n}=F_n=-\overline{F_n}.
\]
For $d\ge 0$, let  
$G_{d}$ 
be the nonlinear Fourier series of the truncated sequence 
 \[
 \left ( F_n 1_{ \{-d\le n\le d \}} \right ) .
 \]
We may write $G_d(z)$ for   $z\in \T$,  using the symmetries of $(F_n)$, recursively as
\begin{equation}\label{g0}
G_{0}(z)=\frac 1{\sqrt{1-F_0^2}} 
\begin{pmatrix}
    1 & F_0 \\
    F_0 & 1
\end{pmatrix},
\end{equation}
\begin{equation}\label{gd}
G_{d}(z)=\frac 1{{1-F_d^2}} 
\begin{pmatrix}
    1 & F_dz^{-d} \\
    F_d z^{d} & 1
\end{pmatrix} 
G_{d-1}(z)
\begin{pmatrix}
    1 & F_dz^d \\
    F_d z^{-d} & 1
\end{pmatrix} .
\end{equation}

Define $X$ and $M$ and recall $Z$ as follows:
\begin{equation}\label{eq:def_X}
         X = \begin{pmatrix}
        0 & 1 \\ 1 & 0
    \end{pmatrix}, \quad
M = 2^{- \frac{1}{2}}\begin{pmatrix}
    1 & 1 \\
    1 & -1
\end{pmatrix} , \quad 
Z = \begin{pmatrix}
        1 & 0 \\ 0 & -1
    \end{pmatrix} .
\end{equation}
Observe that $M^2$ is the identity matrix, that
\begin{equation}
 XM  = 2^{- \frac{1}{2}}
\begin{pmatrix}
    1 & -1 \\
    1 & 1
\end{pmatrix} =MZ,
\end{equation}
and hence also $MZM=X$ and $MXM=Z$.

\begin{lemma}\label{qspnlfs}
 For $x\in [0,1]$ let $\theta$ be the unique number in
$ [0,\frac{\pi}2]$ so that $\cos \theta=x$
and set $z=e^{2 i \theta}$.
We have for every $d\ge 0$ and $U_d$ as in Theorem \ref{main},
$$M U_d(\Psi,x) M= \begin{pmatrix}
e^{id\theta} & 0\\
0 & e^{-id\theta}
\end{pmatrix}
G_{d}(z)
\begin{pmatrix}
e^{id\theta} & 0\\
0 & e^{-id\theta}
\end{pmatrix}.$$
\end{lemma} 
Note that the factor two in the exponent of the definition of $z$ differs from the convention in \cite{Linlin}.

\begin{proof}

We  prove the Lemma by induction on $d$.
For  $k\in \N$, we have
\begin{equation}\label{MeZM}
Me^{i\psi_k Z}M=e^{i\psi_k MZM}
\end{equation}
\begin{equation}\label{ecossin}
=e^{i\psi_k X}=\begin{pmatrix}
\cos(\psi_k) & i\sin(\psi_k)\\
i \sin(\psi_k) & \cos(\psi_k)
\end{pmatrix}
={\cos(\psi_k)}\begin{pmatrix}
1 & i\tan(\psi_k)\\
i \tan(\psi_k) &1
\end{pmatrix}
\end{equation}
\begin{equation}
=\frac 1{\sqrt{1+\tan(\psi_k)^2}}\begin{pmatrix}
1 & i\tan(\psi_k)\\
i \tan(\psi_k) &1
\end{pmatrix}
=\frac 1{\sqrt{1-F_k^2}}\begin{pmatrix}
1 & F_k\\
F_k &1
\end{pmatrix}.
\end{equation}
Applying this with $k=0$ and using \eqref{g0} and \eqref{u0} in the form
\begin{equation}
    Me^{i\psi_0 Z}M=M U_0(\Psi,x) M
\end{equation}
verifies the base case $d=0$ of the induction.

Now let $d\ge 1$ and assume  the induction hypothesis is true for $d-1$.
Noting that similarly as in \eqref{ecossin},
\[
W(x) = e^{i \arccos (x)X} ,
\]
we have 
\begin{equation}MW(x)M =e^{i\arccos(x)MXM}=e^{i\theta Z}=
\begin{pmatrix}
e^{i\theta} & 0\\
0 & e^{-i\theta}
\end{pmatrix}.
\end{equation} 
Hence, with \eqref{MeZM},
\begin{equation}
    \sqrt{1-F_d^2}MW(x)e^{i\psi_d Z}
    =\begin{pmatrix}
e^{i\theta} & 0\\
0 & e^{-i\theta}
\end{pmatrix}\begin{pmatrix}
1 & F_d\\
F_d &1
\end{pmatrix}M
\end{equation}
and 
\begin{equation}
    \sqrt{1-F_d^2} e^{i\psi_d Z}W(x)M
    = M\begin{pmatrix}
1 & F_d\\
F_d &1
\end{pmatrix}\begin{pmatrix}
e^{i\theta} & 0\\
0 & e^{-i\theta}
\end{pmatrix}.
\end{equation}
We obtain with the recursive definition \eqref{ud} and induction hypothesis,
\begin{equation} (1-F_d^2)M U_d(\Psi,x) M
\end{equation}
\begin{equation*} =(1-F_d^2)M e^{i\psi_{d}Z} W(x)M(M U_{d-1}M) M(\Psi,x)W(x)e^{\psi_d Z} M
\end{equation*}
\begin{equation*}
=\begin{pmatrix}
1 & F_d\\
F_d &1
\end{pmatrix}\begin{pmatrix}
e^{id\theta} & 0\\
0 & e^{-id\theta}
\end{pmatrix}G_{d-1}(z)
\begin{pmatrix}
e^{id\theta} & 0\\
0 & e^{-id\theta}
\end{pmatrix}
\begin{pmatrix}
1 & F_d\\
F_d &1
\end{pmatrix}
\end{equation*}
\begin{equation*}
=\begin{pmatrix}
e^{id\theta} & 0\\
0 & e^{-id\theta}
\end{pmatrix}
\begin{pmatrix}
1 & F_dz^{-d}\\
F_d z^{d} &1
\end{pmatrix}
G_{d-1}(z)
\begin{pmatrix}
1 & F_d z^d\\
F_d z^{-d} &1
\end{pmatrix}
\begin{pmatrix}
e^{id\theta} & 0\\
0 & e^{-id\theta}
\end{pmatrix}
\end{equation*}
\begin{equation}
=
(1-F_d^2)
\begin{pmatrix}
e^{id\theta} & 0\\
0 & e^{-id\theta}
\end{pmatrix}
G_{d}(z)
\begin{pmatrix}
e^{id\theta} & 0\\
0 & e^{-id\theta}
\end{pmatrix}.
\end{equation}
This proves the induction step for $d$
and completes the proof of Lemma \ref{qspnlfs}.
\end{proof}

\begin{lemma}\label{upperleft}
Let $d\ge 0$ and set
$$G_d(z)=:\begin{pmatrix}
    a(z)& b(z)\\ -b^*(z) & a^*(z)
\end{pmatrix}.$$
For $x\in [0,1]$, let $\theta$ be the unique number in $[0,\frac \pi 2]$  such that $\cos \theta=x$
and set $z=e^{2 i \theta}$. We have for $d\ge 1$ and $u_d$ as in Theorem \ref{main},
$$ i\Im(u_{d}(\Psi,x))= b(z).$$    
\end{lemma}
\begin{proof}
We use Lemma \ref{qspnlfs} to obtain
$$ U_d(\Psi,x) = M\begin{pmatrix}
e^{id\theta} & 0\\
0 & e^{-id\theta}
\end{pmatrix}
G_{d}(z)
\begin{pmatrix}
e^{id\theta} & 0\\
0 & e^{-id\theta}
\end{pmatrix}M$$
$$  = M\begin{pmatrix}
a(z)z^d  & b(z)\\
-b^*(z)  & a^*(z)z^{-d}
\end{pmatrix}
M.$$
We then compute 
the upper left corner
\begin{equation*}u_{d}(\Psi,x)=\frac 12 
\begin{pmatrix}
1  & 1
\end{pmatrix}
\begin{pmatrix}
a(z)z^d  & b(z)\\
-b^*(z)  & a^*(z)z^{-d}
\end{pmatrix}
\begin{pmatrix}
1  \\ 1
\end{pmatrix}
\end{equation*}
\begin{equation}\label{identifyu}
=\frac 12(a(z)z^d+a^*(z)z^{-d} +b(z)-b^*(z)).
\end{equation}
As $z$ is in  $\T$, the last display becomes
\[\Re(a(z)z^d)+ i \Im (b(z)).\]
As $(F_n)$ is purely imaginary and even,
 the symmetries of the nonlinear Fourier
series imply that $b$ is also purely imaginary.
In particular, \eqref{identifyu} gives Lemma \ref{upperleft}.
\end{proof}

The above proof gives $b(z)=b(z^{-1})$ for $z\in \T$
and that $\Im (u_d(x, \Phi))$ extends to an even function 
in $x\in [-1,1]$.

We note that with this correspondence
between NLFA and QSP established, Theorem 31 and Theorem 5 in \cite{EnergylandscapeofQSP} observe a version of the comparability of $\|F\|_{l^2(\Z)}$ and $\|b\|_{L^2(\T)}$ discussed in the remarks to \eqref{outerplancherel} in the previous section.

\section{Nonlinear Fourier series of summable sequences}
While our focus in this paper is on square summable sequences, we briefly comment on the 
analytically simpler theory of nonlinear Fourier series of elements in the space  $\ell^1(\mathbb{Z})$ of absolutely summable  sequences on $\mathbb{Z}$. The linear Fourier series maps $\ell^1(\Z)$ to the space 
$C(\mathbb{T})$  of continuous functions on $\T$, a closed subspace of 
$L^\infty(\mathbb{T})$. 
The actual
image of 
$\ell^1(\Z)$
under the linear Fourier series is the Wiener algebra  $A(\T)$.
Similar mapping properties are true for the nonlinear Fourier series. 

We first recall the Theorem below of \cite{tsai} for the $L^\infty$ bounds. Consider a metric on $SU(2)$ induced by the operator norm, i.e.,
\begin{equation}
    \dist (T,T') := \| T - T'\|_{op}
\end{equation}
and let $C (\mathbb{T}, SU(2))$ be the metric space of all continuous $G: \mathbb{T}\to SU(2)$ 
with metric defined by
\begin{equation}
    \dist (G,G') :=\sup_{z\in \T} \dist (G(z), G'(z)).
\end{equation}
\begin{theorem}[{\cite[Theorem 2.5]{tsai}}]
\label{tsail1}
    The $SU(2)$ nonlinear Fourier series extends uniquely to a Lipschitz map $\ell^1(\Z)\to C(\mathbb{T}, SU(2))$ with Lipschitz constant at most $3$. 
\end{theorem}

The use of the operator norm is of no particular relevance except possibly for the value of
the Lipschitz constant, because all norms on
the finite dimensional space of $2\times 2$ matrices are equivalent.
Let $b$ and $b'$ be the second entries of the first row of $\overbrace{F}$ and $\overbrace{F'}$, respectively. 
Then Theorem \ref{tsail1} in particular implies
\begin{equation}
    \| b-b'\|_{L^{\infty}} 
    \leq 3 \|F-F' \|_{\ell^1}.
\end{equation}

We next turn to the Wiener algebra $A(\T)$.
Recall that the linear Fourier series is injective from $\ell^1(\Z)$ onto 
$A(\T)$ and the norm $\|.\|_A$ on the Wiener algebra is defined so that the linear Fourier
series is an isometry from $\ell^1(\Z)$ to the Wiener algebra.
Let $A(\mathbb{T},\mathbb{C}^2)$
be the space of pairs $(a,b)$ of functions in $A(\T)$
and let the norm 
of $(a,b)$ be
defined as $\|a\|_A+\|b\|_A$.

\begin{theorem}\label{wienertheorem}
    The $SU(2)$ nonlinear Fourier series is a real analytic map from $\ell^1(\Z)$ to
    $A(\T,\C^2)$.
    
    Let $R\geq 0$, $\|F\|_{\ell^1}, \|F'\|_{\ell^1} \leq R$. If $b$ and $b'$ are the second entries of the nonlinear Fourier series of $F$ and $F'$, respectively, then
    \begin{equation}\label{eq:Lipschitz l1}
    \|b-b'\|_A \leq  e^R \|F-F'\|_{\ell^1} .
    \end{equation}
If additionally $R \leq 0.36$, then
\begin{equation}\label{eq:reverseLipschitz}
        \|F-F'\|_{\ell^1} \leq 2\| b-b'\|_A.
    \end{equation}
\end{theorem}

\begin{proof}

We begin with a finite sequence $F$ and write the nonlinear Fourier series as an ordered product
\begin{equation}
\label{discexp}(a(z),b(z))=\prod_{j=-\infty}^\infty
    (1+ |F_j|^2)^{-\frac{1}{2}}(1,F_jz^j),
\end{equation}
where the non-commutative product is understood in the sense of $j$ increasing from left to right. We decompose $(1,F_jz^j)=(1,0)+(0,F_jz^j)$ and 
apply the distributive law. 
The terms resulting from the distributive law 
are parameterized by increasing sequences
$j_1<\dots <j_n$
of indices, 
for which $F_jz^j$
appears in the term.
Hence we write the right side of \eqref{discexp} as

\begin{equation}\label{discdist}
 C(F)\left( \sum_{n=0}^{\infty} \sum_{j_1<j_2<\dots < j_n} \prod_{k=1}^n (0, F_{j_k} z^{j_k}) \right)
\end{equation}
with
\begin{equation*}
 C(F)=  \prod_{j=-\infty}^\infty (1+ |F_j|^2)^{-\frac{1}{2}} .\end{equation*}

The $n$-th term in the sum of \eqref{discdist}
is diagonal for even $n$ and anti-diagonal for odd $n$. Setting
\begin{equation}\label{texp}
    T_n (F^1,\dots, F^n)(z) := \sum_{j_1<j_2<\dots < j_n} \left(\prod_{\substack{1\le k\le n \\k \text{ is odd}}}F^k_{j_k}z^{j_k} \right) \left ( \prod_{\substack{1\le k\le n \\ k \text{ is even}}}-\overline{F^k_{j_k}} z^{-j_k}\right )  ,
\end{equation}
we obtain
\begin{equation}\label{aexp}
    a =  C(F) \sum_{n=0}^{\infty} T_{2n} (F,\dots, F).
\end{equation}
\begin{equation}\label{bexp}
    b =  C(F)\sum_{n=0}^{\infty} T_{2n+1} (F,\dots, F).
\end{equation}

We will show that both the function $C$ and the multilinear expansions in
\eqref{aexp} and \eqref{bexp}
extend to analytic maps in $\ell^1(\Z)$, thereby proving that
$a$ and $b$ extend to analytic maps in 
the argument $F\in \ell^1(\Z)$.

We first discuss 
$C(F)$. For 
a sequence $H_j$ of non-negative numbers, we have
\begin{equation}
\prod_{j=-\infty}^\infty (1+H_j)=1+\sum_{n=1}^\infty\sum_{j_1<\dots <j_n} \prod\limits_{k=1}^n H_{j_k}\le 
1+\sum_{n=1}^\infty\frac 1{n!}\|H\|_{\ell^1(\Z)},
\end{equation}
which we recognize as a multi-linear expansion with infinite radius of convergence.
As the map $F_j\to F_j\overline{F_j}$
is real analytic
from $\ell^1(\Z)$
to itself and the 
$-\frac 12$-th power is real analytic from $[1,\infty)$ to $(0,1]$, the function $C$
extends to a real analytic map from $\ell^1(\Z)$ to $(0,1]$.

As for $T_n$, taking all $F^j = F$ and  summing \eqref{texp} in absolute value over all permutations of the indices  $j_1$ to $j_n$, the sum separates into a product of sums and  one can estimate for $|z|=1$
\begin{equation}
    |T_n (F,\dots, F) (z)| \leq \frac{1}{n!}  \| F\|_{\ell^1} ^n.
\end{equation}
Thus the multilinear expansions in the expression \eqref{aexp} and \eqref{bexp} of $a$ and $b$ have infinite radius of convergence in $\ell^1(\Z)$ and extend to real analytic maps from $\ell^1$ to $L^{\infty}(\T,\mathbb{C}^2)$. 
Moreover, \eqref{texp} is the
linear Fourier series of the sequence given by 
\begin{equation*}   (\check{T}_n(F^1,\dots, F^n))_j=
 \sum_{\substack{j_1<j_2<\dots < j_n
\\\sum_{k=1}^n -(-1)^kj_k=j}
} \left(\prod_{\substack{1\le k\le n \\k \text{ is odd}}}F^k_{j_k} \right) \left ( \prod_{\substack{1\le k\le n \\ k \text{ is even}}} - \overline{F^i_{j_i}} \right ).
\end{equation*}
Absolutely summing over $j$ as well as over permutations of the indices from $j_1$ to $j_n$ yields that
\begin{equation}
\label{wienersum}
    \|T_n (F^1,\dots, F^n) \|_A \leq \frac{1}{n!} \prod_{j=1}^n \| F^j\|_{\ell^1}.
\end{equation}
Hence the multilinear expansions in \eqref{aexp} and \eqref{bexp} extend to real analytic maps from $\ell^1(\Z)$ to $A(\T)$. The nonlinear Fourier series extends to a real analytic map from $\ell^1(\Z)$ to $A(\T,\mathbb{C}^2)$. This
proves the first statement of Theorem \ref{wienertheorem}.

We turn to the proof of \eqref{eq:Lipschitz l1}.
By an $n$-fold application of the  triangle inequality and \eqref{wienersum} above,
    \begin{equation}
    \label{ttelescope}
    \| T_n (F,\dots, F) - T_n (F',\dots ,F') \|_A 
    \end{equation}
\begin{equation*}
\leq \sum_{j=1}^n \|T_n (F',\dots, F', F-F', F,\dots, F)\|_A \end{equation*}
\begin{equation*}
    \leq \frac{\|F-F'\|_{\ell^1} R^{n-1}}{(n-1)!} ,
\end{equation*}
where in the middle term the difference $F-F'$ occurs in the $j$-th entry.
On the other hand, 
by a telescoping sum as in \eqref{ttelescope}
using that all factors of $C(F)$ are bounded by $1$, we get
\begin{equation}\label{eq:estimateforproduct}
   |C(F)-C(F')| \leq \sum_{j=-\infty}^\infty
    \left|(1+ |F_j|^2)^{-\frac{1}{2}}-
    (1+ |F_j'|^2)^{-\frac{1}{2}}\right|
    \end{equation}
\begin{equation*}
    \le 
    \sum_{j=-\infty}^\infty
    \left|(1+ |F_j|^2)^{\frac{1}{2}}-
    (1+ |F_j'|^2)^{\frac{1}{2}}\right|\le
    \sum_{j=-\infty}^\infty
    |F_j-F'_j|=\|F-F'\|_{\ell^1(Z)}.
\end{equation*}
Thus, by the triangle inequality,  
\begin{equation}\label{eq:516}
    \|b-b'\|_A
 \end{equation}
 \begin{equation*}
    \leq |C(F)-C(F')|\sum_{n=0}^\infty\frac{R^{2n+1}}{(2n+1)!}+C(F')\sum_{n=0}^\infty \frac{\|F-F'\|_{\ell^1}R^{2n}}{(2n)!}
\end{equation*}
\begin{equation}
    \le \left( \sinh (R) + \cosh (R) \right) \|F-F'\|_{\ell^1}
    = e^R \|F-F'\|_{\ell^1}.
\end{equation}
This proves \eqref{eq:Lipschitz l1}.

We turn to the proof of  \eqref{eq:reverseLipschitz}. We first note a lower bound for $C(F)$. We have
\begin{equation}
 \label{logc}
    -2\log(C(F)) = { \sum_{j=-\infty}^{\infty} \log (1+|F_j|^2)) }
 \leq { \sum |F_j|^2 } \leq { \|F\|_{\ell^1}^2 } \leq { R^2 }
\end{equation}
and hence
\begin{equation}
\label{clower} 
C(F)\ge e^{-\frac 1 2  R^2}.
\end{equation}

The key observation is now that $T_1 (F)$ is the linear Fourier series of $F$. We will isolate this term in \eqref{bexp} by the triangle inequality as follows:
\begin{equation}
\label{btriangle}
\|b-b'\|_A + \|\sum_{n=1}^{\infty} C(F) T_{2n+1}(F,\dots,F) - C(F')T_{2n+1}(F',\dots,F')\|_A
\end{equation}
\begin{equation*}
        \geq \|C(F)T_1 (F) - C(F')T_1(F')\|_A
\end{equation*}

\begin{equation*}
    \geq C(F) \|T_1(F)-T_1(F')\|_{A} - |C(F)-C(F')| \|T_1(F')\|_A
\end{equation*}

\begin{equation*}    \geq e^{-\frac{1}{2}R^2} \|F-F'\|_{\ell^1} -  \|F-F'\|_{\ell^1}
\|F\|_{\ell^1} \geq( e^{-\frac{1}{2}R^2}-R) \|F-F'\|_{\ell^1} .
\end{equation*}
Here we have used \eqref{clower} and \eqref{eq:estimateforproduct}.

Estimating the second term on the far left-hand side of \eqref{btriangle} 
analogously to \eqref{eq:516} yields
\begin{equation*}
    \|b-b'\|_A + (e^R-1-R)\|F-F'\|_{\ell^1} \geq (e^{-\frac{1}{2}R^2}-R) \|F-F'\|_{\ell^1}
\end{equation*}
and hence
\begin{equation*}
    \|b-b'\|_A \geq \|F-F'\|_{\ell^1} \left( 1+e^{-\frac{1}{2}R^2} - e^R\right)  ,
\end{equation*}
where the last term in parentheses is larger than $\frac{1}{2}$ for $R\leq 0.36$.
\end{proof}

In \cite{Linlin}, the authors investigate similar inequalities. Up to comparing the constants, Theorem 3, Corollaries 18 and 20 of \cite{Linlin} state the same inequalities as \eqref{eq:Lipschitz l1} and \eqref{eq:reverseLipschitz}. In fact, constants in \cite{Linlin} are better than the ones we obtain. This is due to the fact that we only use the triangle inequality and absorb all the multilinear terms into the first term, whereas \cite{Linlin} carries out a more subtle estimate through the Jacobian of $\sum_{n=0}^{\infty}T_{2n+1}(F,\dots, F)$.

\section{Nonlinear Fourier series of one sided square summable sequences}

In this section, we largely follow \cite{tsai} while giving
a self-contained presentation.

For $1 \leq p \le  \infty$, let $H^p(\D)$ be the classical Hardy space associated to the disc $\D$,
 that is the set of functions $f$ in $L^p(\T)$ which are the linear Fourier series of a sequence supported in $[0,\infty)$.
The linear Fourier series of a Hardy space function $f$
provides an analytic extension of $f$ to $\D$ which has non-tangential limits almost everywhere on $\T$
equal to the function $f$.  We denote the value of the 
extension of $f$ at a point $z\in \D$ by $f(z)$. The anti-Hardy space $H^p (\D^*)$ consists of the functions $f$ on $\T$ for which $f^* \in H^p (\D)$.
The mean value theorem in the form of  \eqref{inftymeanvalue} continues to hold for
functions $a\in H^p(\D^*)$. In particular, values of functions $a$ in $H(\D^*)$ and $H(\D)$ respectively at $\infty$ and $0$, if real,
are the average of the real part of the function on $\T$.

If $f\in H^p(\D)$ is bounded by $1$, then
its extension to $\D$ is bounded by $1$. If $f$ has modulus $1$
almost everywhere on $\T$, then $f$ is called inner.
If $f(0)>0$, then $\log|f|$ is integrable on $\T$ and
\begin{equation}\label{outer}
\int_\T \log|f|\ge \log f(0),
\end{equation}
 and $f$ is called outer if equality holds in \eqref{outer}.
 
If $f\in H^p(\D)$, and $f$ vanishes at $0$, then the imaginary 
part of $f$ is the Hilbert transform $H$ with respect to the circle of the real part of $f$.
If $f\in H^p(\D^*)$, and $f$ vanishes at $\infty$, then the imaginary  part of $f$ is the negative of the Hilbert transform of the real part of $f$. The Hilbert transform has operator norm one
in $L^2(\T)$.

On the Hilbert space $L^2 \left ( \T \right )$, there is the orthogonal projection operator $P_{\D}$ onto $H^2 \left (\D \right )$. We also define 
\[
P_{\D ^*}f = ({P_{\D} ({f}^*)})^* .
\] 
Then $P_{\D^*}$ is the orthogonal projection of $L^2 \left (\T \right )$ onto $H^2 (\D^*)$. Both these operators have operator norm one on $L^2 \left (\T \right )$. 

We refer to \cite{garnett} for these and further details on the theory of Hardy spaces.

Let ${\mathbf L}$ be the set of pairs of 
measurable functions $(a,b)$ on $\mathbb{T}$ such that 
\begin{equation}\label{eq:det_torus_gen} 
aa^*+bb^*=1
\end{equation}
almost everywhere on $\T$ and $a$ is 
in $H^2(\D^*)$ with  $a(\infty)>0$.
We introduce the following metric on ${\mathbf L}$:
\begin{equation}\rho((a,b),(c,d))=\left (\int_{\mathbb T} |a-c|^2 \right )^{\frac 1 2}+
\left (\int_{\mathbb T}{|b-d|^2} \right )^{\frac 12}+|\log(a(\infty))-\log(c(\infty))|.
\end{equation}

The nonlinear Fourier series of a finite sequence is in $\mathbf{L}$. 
Moreover, the metric $\rho$ has the following compatibility with the product
\eqref{leftright} in Theorem \ref{SU2basictransformations}.
\begin{theorem}\label{rhoproduct}
    Let  $F,\tilde{F}$ be finite sequences
with support entirely to the left of the support of some other finite sequences 
$G,\tilde{G}$. Let $(a,b)$, $(\tilde{a},\tilde{b})$, $(c,d)$, $(\tilde{c},\tilde{d})$, be the nonlinear Fourier series of $F,\tilde{F}, G, \tilde{G}$, respectively. Then
\begin{equation}\label{eq:rho_gluing_sequences}
 \rho((a,b)(c,d), (\tilde{a},\tilde{b})
 (\tilde{c},\tilde{d}))
 \le 2\rho((a,b), (\tilde{a},\tilde{b})
 )
 +2\rho((c,d),
 (\tilde{c},\tilde{d})).
\end{equation}
\end{theorem}
\begin{proof}
    Thanks to the support 
properties of the sequences, $bd^*$
vanishes at $\infty $ and 
\[\log(a c-bd^*)(\infty)=\log(a(\infty))+\log(c(\infty))\]
and similarly for the quantities with tildes. This shows the desired inequality for the logarithmic
parts of the metric.

The functions $a,b,c,d$ and their tilde counterparts 
are bounded by $1$ almost everywhere on $\T$. Hence
\[|(ac-bd^*)-
(\tilde{a}\tilde{c}-\tilde{b}\tilde{d}^*)|\]
\[= |(a-\tilde{a})c+\tilde{a}(c-\tilde{c})-(b-\tilde{b})d^*-\tilde{b}(d^*-\tilde{d^*})|\]
\[\le |a-\tilde{a}|+|c-\tilde{c}|+
|b-\tilde{b}|+|d^*-\tilde{d^*}|,\]
and similarly
\[|(ad+bc^*)-
(\tilde{a}\tilde{d}+
\tilde{b}\tilde{c}^*)|
\le |a-\tilde{a}|+|c-\tilde{c}|+
|b-\tilde{b}|+|d^*-\tilde{d^*}|.\]
The desired bounds for the $L^2$ parts of $\rho$ then follow by the triangle inequality.
\end{proof}

Note that as the absolute values of $a$ and $c$ are almost everywhere bounded by $1$ on $\mathbb{T}$,  the metric  
$\rho$ defines the same topology as the metric  on $\mathbf{L}$ in \cite{tsai} using an $L^1$ integral.

Let $\overline{\mathbf H}$ be the set of functions in ${\mathbf L}$
such that $b$ is in $H^2(\D)$.
Let ${\mathbf H}$ be the set of functions in $\overline{\mathbf H}$ such that 
$a^*$ and $b$ have no common inner factor in the sense that if  $a^*g^{-1}$ and $bg^{-1}$ are in $H^2(\D)$ for some inner function $g$ on $\T$, then $g$ is constant. Note that the bar in 
$\overline{\mathbf{H}}$ has the meaning of a closure rather than a complex conjugation. In fact, $\bar{\mathbf{H}}$ is complete.

For a sequence $F$ supported on $[0,\infty)$
define $(a_k,b_k)$ for $k\ge 0$ recursively by
\begin{equation}\label{zerorec} (a_0(z),b_0(z))= (1+|F_0|^2)^{-\frac 1 2}( 1 , F_0)\end{equation}
and for $k>0$
\begin{equation}\label{halfrec}
(a_k(z),b_k(z) )= (a_{k-1}(z),b_{k-1}(z)){(1+|F_k|^2)^{-\frac 12}}(1 , F_k z^{k}).
\end{equation}
This definition coincides with \eqref{nlfsdef} under the identification of $G_k$ with $(a_k,b_k)$ for
sequences supported on $[0,\infty)$.

\begin{theorem}\label{forwardcauchy}
Let $F$ be a sequence in $l^2(\Z)$ with support in $[0,\infty)$. The sequence  $(a_k,b_k)$
as in \eqref{zerorec} and \eqref{halfrec} converges in $\mathbf{L}$ to an element
$(a,b)$ in $\overline{\mathbf{H}}$.
We have \begin{equation}\label{infplancherel}
a(\infty)= \prod_{n\ge 0} (1+|F_n|^2)^{-1/2} .
\end{equation}

\end{theorem}
We call the limit $(a,b)$ in this theorem the nonlinear Fourier
series of $F$. This definition is consistent with the definition of the nonlinear Fourier series  near \eqref{nlfsdef} in the case that $F$ has finite support.
Theorem 
\ref{SU2basictransformations} continues to hold for one sided infinite sequences, by taking limits as in Theorem \ref{forwardcauchy}.

\begin{proof}
We first show that the sequence $(a_k,b_k)$ is Cauchy in 
$\mathbf{L}$. Let $k< l$ and write
\begin{equation}\label{eq:cauchy_seq}
(a_k,b_k)(c,d) = (a_l,b_l)
\end{equation}
where $(c,d)$ is the nonlinear Fourier series
of the sequence $\left ( F_n \mathbf{1}_{\{k+1\leq n \leq l\}} \right )$
and where we have used the multiplicative property \eqref{leftright}
in Theorem \ref{SU2basictransformations}.
We have by Plancherel \eqref{finiteplancherel}, applied to $a_l$, $a_k$, and $c$
\begin{equation}\label{logcauchy}
\log|a_l(\infty)|-\log|a_k(\infty)|=
-\frac 12 \sum_{k<n\le l}\log(1+|F_n|^2) =\log |c(\infty)|.
\end{equation}
The right-hand side tends to zero for $k\to \infty$ because $(F_n)$
is square summable. 

By \eqref{eq:cauchy_seq} we have
$$|a_k-a_l|= |a_k(1-c)+ b_k d^*|\le |1-c|+|d|,$$
$$|b_l-b_k|= |a_k d+b_k (c^*-1)|\le |1-c|+|d|.$$
To control the $L^2$ norms  of the left-hand sides, it suffices to control  the $L^2$ norm of each summand on the right-hand side.
As $1-c(\infty)$ is real and $1-c$ is in $H^2(\D^*)$, the $L^2$ norm of the imaginary part of $1-c$, which equals the negative of the Hilbert transform of the real part of $1-c$, is hence bounded by the $L^2$ norm of its real part, since the Hilbert transform $H$ has norm one acting on $L^2(\T)$
 \cite[Chapter 3]{garnett}. Using this and the  bound $\|c\|_\infty \le 1$, we estimate 
\[ \frac 14\int_\T |1-c|^2 \le \frac 12 \int_\T |\Re(1-c)|^2
\le \int_\T \Re(1-c) =1-c(\infty)\le -\log c(\infty).\]
The latter tends to zero as in \eqref{logcauchy} as $k \to \infty$. 
Moreover, 
\[\int_\T |d|^2=\int_\T 1-|c|^2\le - 2\int_\T \log|c|
\le -2 \log(c(\infty))  ,\]
which also tends to zero as in \eqref{logcauchy}.
Having seen that the sequence $(a_k,b_k)$ is Cauchy with respect to
all three summands in the definition of $\rho$,
it is Cauchy with respect to $\rho$.

As each $(a_k,b_k)$ is in $\overline{\mathbf{H}}$
and $\overline{\mathbf{H}}$ is complete, $(a_k,b_k)$ has a limit in $\overline{\mathbf{H}}$.

\end{proof}

\begin{theorem}\label{layertheorem}
Let $(a,b)\in \overline{\mathbf{H}}$.
There is a unique $y\in \C$ such that there
exists $(c,d)\in \overline{\mathbf{H}}$ satisfying 
\begin{equation}\label{layerstrip}
    (c(z),d(z)z):=(1+|y|^2)^{-1/2} (1,-y)(a(z),b(z))
\end{equation}
for almost all $z\in \T$.
Using this statement, define the functions $(a_n,b_n)$ recursively for $n\ge 0$ by 
\[(a_0,b_0)=(a,b)\]
\begin{equation}\label{layerrecursion}(a_{n+1}(z),b_{n+1}(z)z)=(1+|F_n|^2)^{-\frac 12}
(1,-F_n)(a_n(z),b_n(z)),
\end{equation}
where $F_n$ is the unique number such that $(a_{n+1},b_{n+1})$ is in $\overline{\mathbf{H}}$.
Then the sequence $(F_n)$ is square summable and 
\begin{equation}\label{leplancherel}
a(\infty)\le \prod_{n\ge 0} (1+|F_n|^2)^{-1/2} .
\end{equation}
If $(a,b)\not\in \mathbf{H}$,
then we have strict inequality in \eqref{leplancherel}. 
\end{theorem}
The sequence produced in this theorem is called the layer
stripping sequence of $(a,b)$, \cite{sylvesterwinebrenner}. Layer stripping  is an alternation between a left multiplication by a constant matrix and a shift by $z$, mirroring the alternation between two types of unitary matrices in the 
QSP representation of the nonlinear Fourier series.

\begin{proof}
To see the first statement of Theorem \ref{layertheorem}, note that for each $y \in \C$,
the factor $(1+|y|^2)^{-1/2} (1,-y)$ is in $SU(2)$
and thus the right-hand side of \eqref{layerstrip}
is an almost everywhere $SU(2)$-valued function on $\T$.
Hence $(c,d)$ is $SU(2)$-valued on $\T$.
The matrix product on the right-hand side of \eqref{layerstrip} is
$( a+ {y} b^*, b-ya^*)$.
We have that $a+yb^*$ is in $H^2(\D^*)$ and $b-ya^*$ is in $H^2(\D)$.
Equality in \eqref{layerstrip} requires
$b-ya^*$ to vanish at $0$. There is a unique complex number
$y$ so that this happens, namely 
\begin{equation}\label{eq:layer_stripping_explicit}
y=b(0)/a^*(0)  .
\end{equation}
For this $y$, we note that $b-ya^*$ can be written
as a product of $z$ with an $H^2(\D)$ function, and that, using that $a(\infty)$ is positive and thus equal to $a^*(0)$,
\begin{equation}\label{eq:first_strip} 
a(\infty)+{y}b^*(\infty) = a(\infty)+{y}\overline{b(0)}=
a(\infty) (1+|y|^2)>0.
\end{equation}
We may thus use \eqref{layerstrip} with this $y$ to define $c,d$. In particular, by \eqref{eq:first_strip} and \eqref{layerstrip} we have
$c(\infty)> 0$ .
It follows that $(c,d) \in \bar{\mathbf{H}}$.
We have thus shown existence and uniqueness of $y$
as in the first part of the theorem.

We may define $F_n$ as in the second 
part of the theorem and obtain by induction
\[a_{n+1}(\infty)=a(\infty)\prod_{k=0}^{n}(1+|F_n|^2)^{\frac 12} . \]
As $a_{n+1}(\infty)\le 1$ for all $n$, we obtain
\eqref{leplancherel}.

Now assume $(a,b)\not \in H$. Then there is a non-constant inner function $g$ such that $\tilde{a} ^* =a ^* g^{-1}$ and $\tilde{b}=bg^{-1}$ are in $H^2(\D)$. Multiplying by a number of modulus one, we may assume
that $g(0)\geq 0$, and since $g$ is not constant then by the maximum principle we have 
$g(0)<1$. Multiplying \eqref{layerrecursion} by $((g^*)^{-1}, 0)$ from the right, one obtains inductively that 
the layer stripping sequence $\tilde{F}$ of $(\tilde{a},\tilde{b})$ is the same as that of $(a,b)$. 
But
\[ \tilde{a}(\infty)> g^*(\infty)\tilde{a}(\infty)=a(\infty).\]
Applying \eqref{leplancherel} to $(\tilde{a},\tilde{b})$
then gives \eqref{leplancherel} for $(a,b)$ with strict inequality.

\end{proof}
\begin{theorem}\label{halflinethm}
Let $F$ be a sequence in $l^2(\Z)$ with support in $[0,\infty)$.
Then the layer stripping sequence of $\overbrace{F}=(a,b)$ is $F$. Moreover, $(a,b)$ is in $\mathbf{H}$.
Conversely, if $(a,b)$ is any element in $\mathbf{H}$, then its layer stripping
sequence is square summable and  
$(a,b)$ is the nonlinear Fourier series of this layer stripping sequence.
\end{theorem}

\begin{proof}
To prove the first part of the theorem, assume to get a contradiction that there is a sequence $F$ such that the layer stripping sequence of
$\overbrace{F}$ is not equal to $F$. Let $n$ be the minimal index such that
the $n$-th term of $F$ differs from the $n$-term of the layer stripping sequence of
$\overbrace{F}$.
We may assume that $n$ is minimal among all hypothetical 
counterexamples to the first
statement of the theorem.

Let $\tilde{F}$ be the sequence supported on $[1,\infty)$ which coincides with $F$ on 
$[1,\infty)$.
Let $(a_k,b_k)$ and $(\tilde{a}_k,\tilde{b_k})$ be the respective sequences defined by the recursion \eqref{zerorec}, \eqref{halfrec}. 

By induction, we obtain for $k\ge 0$
\begin{equation}
(1+|F_0|^2)^{-\frac 12} (1,-F_0)(a_k,b_k)=(\tilde{a}_k,\tilde{b}_k).
\end{equation}
Taking a limit as $k\to \infty$ with
Theorem \ref{forwardcauchy}, we obtain for the nonlinear Fourier series $(a,b)$ and $(\tilde{a},\tilde{b})$
\begin{equation}(1+|F_0|^2)^{-\frac 12} (1,-F_0)(a,b)=
(\tilde{a},\tilde{b}).
\end{equation}
By Theorem \ref{finitethm}, $\tilde{b}_k(0)=0$, for all $k\ge 0$, and by taking limits, as evaluation at $0$ is continuous in $H^2(\D)$, we have
$\tilde{b}(0)=0$. 
Hence there is $d\in H^2(\D)$ such that $d(z)z=\tilde{b}(z)$.
Set $c=\tilde{a}$, then  
\begin{equation}(1+|F_0|^2)^{-\frac 12} (1,-F_0)(a(z),b(z))=
(c(z),d(z)z).
\end{equation}
By the uniqueness part of Theorem
\ref{layertheorem}, we have that $F_0$
is the zeroth entry of the layer stripping sequence
of $(a,b)$. In particular, $n\ge 1$.
By definition, the later terms of the layer stripping sequence
of $(a,b)$ are those of the layer
stripping sequence of $(c,d)$. But $(c,d)$ is the nonlinear Fourier series of the sequence $H_{n}=\tilde{F}_{n+1}$ by Theorem \ref{SU2basictransformations}. It follows that
the $(n-1)$-st term of the 
layer stripping sequence of $(c,d)$ does not coincide with $H_{n-1}$.
This contradicts the minimality of $n$.

Thus we have shown that the layer stripping sequence of $(a,b)$ is equal to $F$.
As we have the Plancherel identity \eqref{infplancherel},
we observe that $(a,b)\in \mathbf{H}$
by Theorem \ref{layertheorem}.

We turn to the second part of the Theorem.
Let $(c,d)\in \mathbf{H}$, let $F$ be its layer stripping sequence and let $(c_k,d_k)$ be the corresponding sequence as defined in Theorem \ref{layertheorem}. By \eqref{leplancherel}, the sequence $F$ is square summable. Let $(a,b)$ be the nonlinear Fourier series of $F$ and let $(a_k,b_k)$
be as in \eqref{zerorec}, \eqref{halfrec}. By induction, we show that for each $k\ge 0$ we have 
\begin{equation}\label{injective}
    (a_k(z),b_k(z))(c_{k+1}(z),d_{k+1}(z)z^{k+1})=(c(z),d(z)).
\end{equation}
Namely, for $k=0$, both sides of the equation are equal to
\[(1+|F_0|^2)^{-\frac 12} (1,F_0)(c_{1}(z),d_{1}(z)z).\]
For $k\ge 1$, the left-hand side of \eqref{injective} is
\begin{equation}\label{injective3}(a_{k-1}(z),b_{k-1}(z))
(1+|F_k|^2)^{-\frac 12} (1,F_k(z)z^k)(c_{k+1}(z),d_{k+1}(z)z^{k+1})\end{equation}
 by definition of $(a_k,b_k)$.
Algebraic verification analogous to a shift show that we may replace
$z^k$ and $z^{k+1}$ in \eqref{injective3}
by $1$ and $z$, respectively.
Hence, by the definition of the sequence $(c_k,d_k)$, the expression \eqref{injective3} is equal to
\begin{equation}
\label{eq:strip_with_finite_approx}
    (a_{k-1}(z),b_{k-1}(z))
(c_{k}(z),d_{k}(z)z^{k}),
\end{equation}
which by the induction hypothesis is the right-hand side of \eqref{injective}.
This completes the induction step and proves \eqref{injective}.

Dividing by the left factor of \eqref{eq:strip_with_finite_approx}, we obtain from \eqref{injective}

\[(c_k(z),d_k(z)z^k)=(a_k^*(z),-b_k(z))(c,d).\]
The entries of the matrix on the right-hand side converge in $L^2(\T)$
because the entries of $(a_k^*,b_k)$ converge and 
such convergence is preserved under 
multiplication by a bounded measurable function. Therefore, the entries of the left-hand side converge in $L^2 (\T)$ as well.

 The limit of $d_k(z)z^k$ is in the closed subspace  $H^2(\D)$. The linear Fourier coefficients of this limit vanish  because the $k$ leading 
 Fourier coefficients of 
 $d_k(z)z^k$ vanish and the map taking a function to any individual Fourier coefficient is continuous in the space $H^2(\D)$. 
 Hence $d_k(z)z^k$ converges to zero. Also $c_k ^*$ converges in $H^2 (\D)$. The limit is an inner function $g$, because  $|c_k|=\sqrt{1-|d_k|^2}$ converges pointwise almost everywhere to $1$. Taking limits in \eqref{injective}
with Theorem \ref{rhoproduct} gives
\[(a,b)(g^*,0)=(c,d)\]
and thus $ag^*=c$ and $bg=d$. As $(c,d)\in \mathbf{H}$, we have by definition 
that $g$ is constant. This constant is positive and hence is equal to one.
It follows that  $(c,d)=(a,b)$, which proves the second part of Theorem \ref{halflinethm}.

\end{proof}

\section{Nonlinear Fourier series of square summable sequences.}

In this section, we adapt
arguments in \cite{TaoThiele2012}
to the $SU(2)$ setting. Unlike \cite{TaoThiele2012}, we need to assume an effective bound on $b$ and assume that $a$ is outer. Some of the complex analytic tools need adjustments. We give a self-contained presentation.

Recall $\mathbf{H}$ and define $\mathbf{H}_0^*$ to be the set of 
$(a,b)$ in ${\mathbf L}$ such that 
$(a, b^*)$ is in $\mathbf{H}$
and $b^*(0)=0$. By the shift and mirror symmetries
of Theorem \ref{SU2basictransformations},
the results of the previous section apply in symmetric form. In particular, $\mathbf{H}_0^*$ is the space
of nonlinear Fourier series 
of sequences in $l^2(\Z)$ 
supported on $(-\infty, -1]$.

We split a sequence $F$ in $l^2(\mathbb{Z})$,  as
$F_-$ + $F_+$, where $F_-$ is supported in $(-\infty, -1]$  
and $F_+$ is supported in $[0,\infty)$.
Let $(a_-,b_-)$ in $\mathbf{H}_0^*$ and $(a_+,b_+)$ in $\mathbf{H}$ be the nonlinear Fourier series of $F_-$ and $F_+$, respectively. Then we define  $(a,b)$ 
almost everywhere on $\T$ by
\begin{equation}\label{riemannhilbert} (a,b):=(a_-,b_-)(a_+,b_+).\end{equation}
As a product of $SU(2)$ matrices almost everywhere, 
$(a,b)$ is in $SU(2)$ and thus entry-wise bounded 
by one almost everywhere. The identity
\begin{equation} a= a_- a_+ -b_- b_+^*
\end{equation}
shows that $a$ has an analytic extension to $\mathbb{D}^*$
with
\begin{equation}
\label{triplea}
a(\infty)= a_-(\infty) a_+(\infty)  .
\end{equation}
Hence $a$ is in $H^2(\D^*)$  with $a(\infty)>0$ and we have $(a,b)\in \mathbf{L}$.

We define $(a,b)$ 
as in \eqref{riemannhilbert} to be the nonlinear Fourier series of $F$.
By Theorems \ref{rhoproduct} and \ref{forwardcauchy}
and the symmetries of Theorem \ref{SU2basictransformations},  
$(a,b)$ is the limit as $k\to \infty$ in $\mathbf{L}$
of the nonlinear Fourier  series of the truncations
of $F$ to the intervals $[-k,k]$.
The properties in Theorem \ref{SU2basictransformations} 
continue to hold for this extension of the  definition of 
nonlinear Fourier series. We also see with \eqref{infplancherel}
 and  \eqref{triplea} that
\begin{equation}\label{fullplancherel}
a(\infty) =  \prod_{n\in \Z} (1+|F_n|^2)^{-1/2} .
\end{equation}

Let $\mathbf{B}$ be the subspace of $\mathbf{L}$ of 
all $(a,b)$ such that 
\begin{align}\label{eq:a_bded_below}
    \inf_{z\in \mathbb{D}^*}|a(z)|^2> \frac 1{2}.
\end{align} 
For a function $a$ satisfying \eqref{eq:a_bded_below}, there is a holomorphic  branch of $\log(a^*)$ on $\D$ with nontangential limits coinciding with $\log(a^*)$ on the boundary. By the mean value theorem for the real part of $\log(a^*)$,  $a^*$ is outer on $\D$.

We embed $\mathbf{B}$ into the Hilbert space $ \mathcal{H} \equiv L^2 \left( \T \right ) \oplus L^2 \left( \T \right )$, written as column vectors,  with the norm
\[
\left \| \begin{pmatrix}
    a \\
    b
\end{pmatrix} \right \|_{\mathcal{H}}  =  \sqrt{\left \|a \right \|_{L^2 \left (\T \right )} ^2 + \left \|b \right \|_{L^2 \left (\T \right )} ^2 } .
\]
For $(a,b)$ and
$(c,d)$ in  $\mathbf{B}$, the metrics defined by $\mathcal{H}$ and $\rho$ are equivalent,
\[\frac 1{8}\rho((a,b),(c,d)) \le \left \|\begin{pmatrix}
    a\\
    b
\end{pmatrix}-
\begin{pmatrix}
    c\\
    d
\end{pmatrix}\right \|_{\mathcal{H}}\leq \rho((a,b),(c,d)).\]
Indeed, the second inequality follows directly from the definition of $\rho$, while the first follows from the additional observation that for outer functions $a^*$ and $c^*$,
\[\left | \log|a(\infty)|-\log|c(\infty)| \right |= 
\left | \int_\T \log|a|-\log|c| \right | \le 2\int_\T|a-c|\le 2\|a-c\|_{L^2(\T)},\]
which used an elementary inequality for the logarithm in the domain $[\frac 1 2, 1]$.

\begin{theorem}\label{exista}
    For each complex valued measurable function $b$ on $\mathbb{T}$ with
    \begin{equation}\label{bbound}
    \esssup_{z\in \mathbb{T}}|b(z)|^2< \frac 1{2},
    \end{equation}
    there is a unique measurable function $a$ on $\mathbb{T}$ such that
    $(a,b)\in \mathbf{B}$.
    \end{theorem}
\begin{proof}
To see existence of $a$, let 
\[
M (z) \equiv \log \sqrt{1 - \left | b (z) \right |^2} 
\]
for almost every $ z \in \T$. By \eqref{bbound}, $M$
is real and integrable on $\T$. Then $M-iHM$ with $H$ the Hilbert transform on $\T$ has an analytic extension to $\D^*$.
Define
\begin{align}\label{eq:b_to_a}
a :=e^{M-iHM}  ,
\end{align} which is in $H^2(\D^*)$  and satisfies
\[
\left |a (z)\right | = e^{M} = \sqrt{1 - \left | b(z) \right |^2} 
\]
for almost every $z\in \T$. Also $a^{-1}$ has analytic extension to $\D^*$ and is bounded by $2^{\frac 1 2}$. It follows that
 \[\inf\limits_{z \in \D^*} \left | a \left (z \right ) \right |^2 > \frac{1}{2}  .
\] Hence $(a,b) \in \mathbf{B}$.

To see uniqueness of $a$, let $\tilde{a}$ be another function as claimed in the theorem. Then $\tilde{a} a^{-1}$
and its reciprocal  are analytic in the disc $\D^*$ with
boundary values of modulus one almost everywhere.
Hence both are bounded by $1$ on the disc and thus 
of modulus one and are hence constant. This constant is positive
at $\infty$ and thus $1$. This proves uniqueness.

\end{proof}

\begin{theorem}\label{inverse}
    For each $(a,b)\in \mathbf{B}$, there are unique $(a_+,b_+)\in \bH$ and $(a_-,b_-)\in \bH_0^*$ such that we have the Riemann Hilbert type factorization  
\begin{equation}\label{eq:RH_factorization_1}
(a_{-}, b_{-})(a_{+}, b_{+}) = (a,b)
\end{equation}
almost everywhere on $\T$.
Moreover, there is a unique
    $F\in l^2(\Z)$ whose nonlinear Fourier series is $(a,b)$.
\end{theorem}

\begin{proof}
Existence and uniqueness of the factorization 
\eqref{eq:RH_factorization_1} shows existence and uniqueness of $F$ by
the one sided Theorem \ref{halflinethm} and the definition 
\eqref{riemannhilbert} of the nonlinear Fourier series on
$\ell^2(\Z)$.
It therefore suffices to show existence and uniqueness of the factorization.

We first discuss uniqueness and begin by deducing necessary conditions on the factors in \eqref{eq:RH_factorization_1}.
Multiplying by the inverse of the matrix $(a_+,b_+)$ from the right
in \eqref{eq:RH_factorization_1}, we obtain

\begin{align}\label{aminusbminus}
(a_{-},  b_{-}) = (a , b)  (a_{+} ^* , -b_{+} ) = (a a_{+} ^* + b b_{+} ^* , - a b_{+} +  a_{+} b)  .
\end{align}
In particular, the second component of this identity reads as
\begin{align}\label{eq:b_linear_reln}
 b_{-} = - a b_{+} +  a_{+} b  .   
\end{align}
Because $|a|$ is bounded below almost everywhere on $\T$, we can divide by $a$ to get
\begin{equation}\label{rh1}
b_{+}=-\frac{b_{-}}{a}  + \frac{b}{a} a_{+}  . 
\end{equation}

The term $b_+$  is in $H^2(\D)$. 
The term $\frac{b_-}{a}$ has an analytic extension to $\D^*$ and hence is in $H^2(\D^*)$ because $|a|$ is bounded below by $2^{-\frac{1}{2}}$. Moreover, $\frac{b_-}{a}$  vanishes at $\infty$.
Acting on \eqref{rh1} by the Cauchy projection $P_{\D}$ yields
\begin{align}\label{eq:b+_projection}
 b_{+}= P_{\D} \left ( \frac{b}{a} a_{+} \right )  .
\end{align}
We similarly rewrite the identity for the first component of \eqref{aminusbminus} as
\[
a_{+} ^*=\frac{a_{-}}{a} - \frac{b}{a} b_{+} ^* , 
\]
and applying $P_{\D}$ yields
\begin{align}\label{eq:a+_projection}
a_{+} ^* = \frac{1}{a _{+} (\infty)} - P_{\D} \left ( \frac{b}{a} b_{+} ^* \right ) .
\end{align}
Here we used that ${a_-} {a}^{-1}$ has analytic extension to $\D^*$ and applying $P_\D$ to it gives the constant term in the linear Fourier expansion, which is equal to
${a_-(\infty)}{a(\infty)}^{-1}$, which is positive and equal to ${a_+(\infty)}^{-1}$ by  
 \eqref{triplea}.

Motivated by \eqref{eq:b+_projection} and \eqref{eq:a+_projection},  we consider the mapping 
\begin{align}\label{eq:contraction_mapping}
(A,B) \mapsto \left  ( \left ( 1 - P_{\D}\left ( \frac{b}{a}B^* \right ) \right ) ^*,  P_{\D} \left ( \frac{b}{a} A \right ) \right )   ,
\end{align}
which is a contraction on $\mathcal{H}$ because $P_{\D}$ is a projection and by \eqref{eq:a_bded_below}, 
\[
\esssup\limits_{z \in \T} \left | \frac{b \left (z \right )}{a \left (z \right )} \right | = \sqrt{\esssup\limits_{z \in \T} \frac{1}{\left |a \left (z \right ) \right |^2} - 1}< 1  , 
\]
where we used that $a$ has limits almost everywhere on $\T$.
Thus \eqref{eq:contraction_mapping} has a unique fixed point $(A,  B)$ by Banach's fixed point theorem. 
Multiplying \eqref{eq:b+_projection} and \eqref{eq:a+_projection}
by $a_+(\infty)$ shows
\begin{equation}\label{eq:fixed_point_+}
(A,B) = a_+ (\infty) ( a_+, b_+) \,  
\end{equation}
is the fixed point of \eqref{eq:contraction_mapping}
 and therefore the right side of \eqref{eq:fixed_point_+} 
is uniquely determined.
Evaluating at infinity gives
\begin{equation}\label{Aaa}
A(\infty)= a_+(\infty)^2.
\end{equation}  
Identity \eqref{Aaa} is necessary and thus the positive value
$a_+(\infty)$ is unique. Dividing the necessary \eqref{eq:fixed_point_+} by this unique number shows that $(a_+,b_+)$ is unique. And by \eqref{eq:RH_factorization_1},  $(a_{-}, b_{-})$ is also unique.

For existence of a factorization \eqref{eq:RH_factorization_1}, again consider the map in \eqref{eq:contraction_mapping} on 
$\mathcal{H}$, and let $(A,B) \in \mathcal{H}$  be the unique solution.
We claim that 
\begin{align}\label{eq:fn_det_fixed_pt}
M \equiv A A^* + B B^*
\end{align}
is constant on $\T$. Clearly, $M$ is real on $\T$, so it suffices to show that $M$ is in $H^1 \left ( \D^* \right )$. Indeed, this will ensure the linear Fourier coefficients of $M$ are supported on $(-\infty, 0]$, while the conjugate antipodal symmetry of the Fourier coefficients of any real valued function ensures the negative Fourier coefficients of $M$ vanish just as its positive coefficients do. We use the fact that $(A,B)$ is a fixed point of \eqref{eq:contraction_mapping} to write
\begin{align}
M = A \left [ 1 - P_{\D}\left ( \frac{b}{a}B^* \right )  \right ]  + B^*  P_{\D} \left ( \frac{b}{a} A \right ) \label{eq:expand_det_fn_fixed_pt},
\end{align}
which after adding and subtracting  $AB^*ba^{-1}$ gives
\begin{align}
M= A \left [ 1 + (\Id- P_{\D})\left ( \frac{b}{a}B^* \right )  \right ]  - B^*  \left (\Id- P_{\D} \right ) \left ( \frac{b}{a} A \right ) \label{eq:M_expanded}  .
\end{align}
We recognize $\Id - P_{\D}$ as the projection operator onto $H^2 _0 \left (\D^* \right )$, i.e., the set of functions in $H^2 \left (\D^* \right )$ with vanishing zeroth Fourier coefficient. 
As the fixed point equation shows that $A\in H^2(\D^*)$ and $B\in H^2(\D)$, then $M$ is a sum of products of functions in $H^2 (\D^*)$, which must then belong in $H^1 (\D^*)$.

By \eqref{eq:M_expanded}, we also have $M\left ( \infty \right )$ equals
\[ A \left ( \infty \right )\left [ 1 + (\Id- P_{\D}) \left ( \frac{b}{a}B^* \right ) \left ( \infty \right )  \right ]  - B^* (\infty) \left (\Id- P_{\D} \right ) \left ( \frac{b}{a} A \right )(\infty) =A(\infty)  .
\]
We can also write
\begin{align} \label{eq:M_equal_Ainfty} 
  A (\infty) = M =
\int\limits_{\T}  M = \int\limits_{\T} \left |A\right |^2 + \left |B\right |^2 \geq 0   ,
\end{align}
where equality holds in the last step if and only if $\left |A \right | = \left |B \right | = 0$ almost everywhere on $\T$. However $(A,B) = (0, 0)$ is not the fixed point of \eqref{eq:contraction_mapping}, hence $A(\infty) > 0$.

Define normalized versions of  $A$ and $B$, i.e.,
\begin{align}\label{abnormalized}
    a_{+} \left (z \right) \equiv \frac{A(z)}{A(\infty)^{\frac{1}{2}}}  , \quad  
    b_{+} \left (z \right) \equiv \frac{B(z)}{A(\infty)^{\frac{1}{2}}}  ,
\end{align}
so that  $a_{+} \in H^2 \left ( \D^* \right )$, $b_+ \in H^2 \left( \D \right )$ satisfy 
\[
a_{+} a_{+} ^* + b_{+} b_{+} ^*  = \frac{M}{A \left (\infty \right )} = 1 
\]
and thus $(a_+,b_+)\in \overline{\mathbf{H}}$.

Now we define 
\[
\left (a_{-}, b_{-} \right ) \equiv \left ( a,b \right ) \left ( a_{+} ^*, -b_{+} \right )
\]
Because $(a_{-}, b_{-})$ is the product of matrices in $SU(2)$, we have
\[
a_{-} a_{-} ^* + b_{-} b_{-} ^* = 1 \text{ on } \T  .
\]
We claim that $\left (a_{-}, b_{-} \right ) \in \overline{\bH}_{0} ^*$, that is $a_-,b_- \in H^2 (\D^*)$ and $b_-(\infty)=0$. 

Indeed, because $(A, B)$ is a fixed point of  \eqref{eq:contraction_mapping}, then $a_{-}$ equals
\begin{align*}
    a  a_{+} ^* + b b_{+} ^* = a \left ( \frac{1}{a_{+}  (\infty)} - P_{\D}\left ( \frac{b}{a}b_{+} ^* \right ) \right ) + b b_{+} ^* =  a \frac{1}{a_{+}(\infty)} + a\left (\Id - P_{\D} \right) \left ( \frac{b}{a}b_{+} ^* \right )  ,
\end{align*}
which is clearly an element of $H ^2 \left (\D^* \right )$ with constant term $ \frac{ a(\infty)}{a_{+} (\infty)} $,
which is positive as both numerator and denominator  are positive.
Using  \eqref{eq:contraction_mapping} again, we have
\[
b_{-} = -a b_{+} + b a_{+} = - a  P_{\D} \left ( \frac{b}{a} a_+ \right ) + b a_{+} = a\left (\Id - P_{\D} \right ) \left ( \frac{b}{a} a_+ \right )  ,  
\]
which again is in $H^2 \left (\D^* \right )$ and has constant term $b_-(\infty)=0$. Thus we see that $(a_-,b_-)$ is in $ \overline{\bH}^*_0$.

To check that  $(a_+, b_+) $ and $(a_{-}, b_{-})$
are indeed in $ \bH$ and
$ \bH_{0} ^*$, it remains to show that $a_{+} ^*$ and $b_+$ share no common nontrivial inner factor $g$ on $\D$, and likewise for  $a_{-} ^*$ and $b_{-}^*$.  Suppose first that $g$ is an inner function such that ${a_{+} ^*}{g^{-1}}$ and ${b_{+}}{g^{-1}}$ are both in $H^2 (\D)$. 
Then by \eqref{eq:RH_factorization_1}, we have
\begin{equation}\label{eq:common_inner_fn}
{a ^*}{g^{-1}} = a_{-} ^* {a_{+} ^*}{g^{-1}} - b_{-} ^* {b_{+}}{g^{-1}}   ,
\end{equation}
which is an $H^2 (\D)$ function. Thus $g$ is an inner factor of $a^*$. But $a^*$ is outer as observed near \eqref{eq:a_bded_below}. This implies that $|g(0)|=1$ because otherwise
\begin{equation}\label{eq:inner_factor_cst}
    \log|a^*(0)|<\log|a^*g^{-1}(0)|\le \int_\T \log|a^*g^{-1}|=\int_\T \log|a^*|,
\end{equation}
a contradiction. By the maximum principle,  $g$ is constant.
Hence $a_{+} ^*$ and  $b_{+}$ share no common inner factor and so
$(a_+,b_+)\in \bH$. Similar reasoning with \eqref{eq:common_inner_fn} and \eqref{eq:inner_factor_cst} shows $a_{-}^*$ and $b_{-}$ share no common inner factor, and hence 
$(a_{-}, b_{-})$
is in
$ \bH_{0} ^*$.

\end{proof}

Given $\epsilon > 0$, let $\mathbf{B}_{\epsilon}$ be the subset of all
 elements $(a,b)$ of $\mathbf{B}$ which satisfy
 \begin{equation}\label{definebe}
  \inf\limits_{z \in \D} \left | a (z) \right |  \geq 2^{-\frac{1}{2}} + \epsilon .
 \end{equation}
\begin{lemma}\label{lemma:best_cst_eta}
Let $\epsilon \in (0, 1 -2^{-\frac 1 2})$ and let $(a,b) \in \mathbf{B}_{\epsilon}$. Then for $\cst \equiv 3^{\frac 3 2}2^{-1}$ we have
\begin{equation}\label{abbound}
\esssup_{z \in \T} \frac{\left |b (z) \right |}{ \left |a (z) \right |}  \leq 1 -  \cst  \epsilon  . 
\end{equation}
\end{lemma}
\begin{proof}
For $x$ in the interval $I \equiv (2^{-\frac{1}{2}}, 1)$, define the positive function
\[ 
f(x) \equiv \sqrt{x^{-2} - 1} .   
\]
Then
\[
f'(x) = \frac{- x^{-3}}{\sqrt{x^{-2} - 1}} = \frac{-1}{\sqrt{x^{4} - x^6}}
\]
achieves its maximum on $I$ when $x^4 - x^6$ achieves its maximum.
Because  
\[
(x^{4} - x^6)' =4 x^3 - 6 x^5 = -6 x^3\left ( x - 2^{ \frac 1 2} 3^{- \frac 1 2} \right )   \left (x + 2^{ \frac 1 2} 3^{- \frac 1 2} \right ) \, , 
\]
then $x^{4} - x^6$, and hence $f'(x)$, achieves its maximum on $I$ at $x=2^{ \frac 1 2} 3^{- \frac 1 2}$ .

Now, let $(a,b) \in \mathbf{B}_{\epsilon}$. Using first \eqref{eq:det_torus_gen} and then our assumption \eqref{definebe}, we write
\begin{equation*}
\esssup_{z \in \T} \frac{\left |b (z) \right |}{ \left |a (z) \right |} = \esssup_{z \in \T} \sqrt{\left |a (z) \right |^{-2} -1} \leq \sqrt{\left ( 2^{-\frac 1 2} +\epsilon \right )^{-2} -1}  = f (2^{- \frac 1 2 } +\epsilon). 
\end{equation*}
By the mean value theorem, there exists $\xi \in (2^{- \frac 1 2 }, 2^{- \frac 1 2 } + \epsilon)$ for which this equals
\[
f(2^{- \frac 1 2}) +f' (\xi) \epsilon \leq f( 2^{ - \frac 1 2}) +f' \left ( 2^{ \frac 1 2} 3^{- \frac 1 2} \right )\epsilon = 1 - 3 ^{\frac{3}{2} }2^{-1} \epsilon \, . 
\]
\end{proof}

\begin{theorem}\label{thm:Lipschitz}
  If we endow both $\mathbf{B}_{\epsilon}$ and $\bH$ with the metric from $\mathcal{H}$, then
 the map sending $(a,b)\in \mathbf{B}_\epsilon$ to the right factor $(a_+,b_+)\in \bH$ of  \eqref{eq:RH_factorization_1} in Theorem \ref{inverse} is Lipschitz  with constant at most $(2^{\frac{5}{2}}+4) (\cst \epsilon)^{-\frac 32}$, where $\cst$ is the constant in \eqref{abbound}. 
\end{theorem}

\begin{proof}

For the fixed point $(A,B)$ of the map \eqref{eq:contraction_mapping} 
we obtain 
\begin{equation}\label{xbootstrap}
\left \|  \begin{pmatrix} A\\ B\end{pmatrix} \right \|_{\mathcal{H} } \leq \left \| \begin{pmatrix}
   0 & -P_{\D ^* } \frac{b^*}{a^*}\\
   P_{\D} \frac{b}{a} & 0
\end{pmatrix} \right \|_{\mathcal{H} \to \mathcal{H}}\left \|  \begin{pmatrix} A\\ B\end{pmatrix} \right \|_{\mathcal{H} } +  \left \| \begin{pmatrix}
    1\\
    0
\end{pmatrix} \right \|_{\mathcal{H}}\end{equation}
\begin{equation*}\leq  \left ( 1- \cst {\epsilon} \right ) \left \|  \begin{pmatrix} A\\ B\end{pmatrix} \right \|_{\mathcal{H} } + 1 ,
\end{equation*}
where we used \eqref{abbound} and the fact that $P_{\D}$ and $P_{\D^*}$ have operator norm $1$ on $L^2 (\T)$. Collecting
the norms of $(A,B)^T$ on the left hand side and dividing by $\cst \epsilon$, we obtain
\begin{align}\label{eq:norm_AB_estimate} 
\left \|  \begin{pmatrix} A\\ B\end{pmatrix} \right \|_{\mathcal{H} } \leq \frac{1}{\cst \epsilon}  .
\end{align}
By the mean value property and the Cauchy-Schwarz inequality, we obtain further
\begin{align}\label{eq:A_infty_estimate}
 \left |A \left (\infty \right ) \right | = \left | \int\limits_{\T} A \right | \leq \left \| A \right \|_{L^2 \left (\T\right )} \leq 
\frac{1}{ \cst \epsilon}  .   
\end{align}

To see that the map from $\mathbf{B}_{\epsilon}$ to the fixed point of \eqref{eq:contraction_mapping} is Lipschitz, let $(a,b), (c,d) \in \mathbf{B}_{\epsilon}$, and let $(A,B)$ and $(C,D)$ be the respective fixed points, i.e.,
\begin{align}
\begin{pmatrix}
    A\\
    B
\end{pmatrix} &= \begin{pmatrix}
    1\\
    0
\end{pmatrix} + \begin{pmatrix}
   0 & -P_{\D ^* } \frac{b^*}{a^*}\\
   P_{\D} \frac{b}{a} & 0
\end{pmatrix} \begin{pmatrix}
    A\\
    B
\end{pmatrix}, \label{eq:A_fixed_point_map} 
\end{align}
\begin{align}
\begin{pmatrix}
    C\\
    D
\end{pmatrix} &= \begin{pmatrix}
    1\\
    0
\end{pmatrix} + \begin{pmatrix}
   0 & -P_{\D ^* } \frac{d^*}{c^*}\\
   P_{\D} \frac{d}{c} & 0
\end{pmatrix} \begin{pmatrix}
    C\\
D
\end{pmatrix}  . \label{eq:C_fixed_point_map}
\end{align}
We subtract the second equation from the first to get an equation for
\[X=\begin{pmatrix}
    A\\
    B \end{pmatrix} - 
\begin{pmatrix}
    C\\
    D \end{pmatrix},\]
namely
\begin{align}\label{eq:X_diff_contraction}
    X=
 \begin{pmatrix}
   0 & -P_{\D ^* } \frac{b^*}{a^*}\\
   P_{\D} \frac{b}{a} & 0
\end{pmatrix} 
X +   \begin{pmatrix}
   P_{\D ^* }\left (  \frac{d^*}{c^*} -\frac{b^*}{a^*} \right )D\\
   P_{\D} \left ( \frac{b}{a} - \frac{d}{c} \right )C
\end{pmatrix}    .
\end{align}
This equation is analogous to \eqref{xbootstrap}, and the same 
bootstrapping argument as there leading to \eqref{eq:norm_AB_estimate}, combined with the fact that $P_{\D}, P_{\D^*}$ have operator norm $1$, gives 
\begin{equation}\label{firstx}\|X\|_{\mathcal{H}}\le \frac1{\cst \epsilon}\left\|
\begin{pmatrix}
   P_{\D ^* }\left (  \frac{d^*}{c^*} -\frac{b^*}{a^*} \right )D\\
   P_{\D} \left ( \frac{b}{a} - \frac{d}{c} \right )C
\end{pmatrix} 
\right\|_{\mathcal{H}}
 \leq  \frac 1{\cst \epsilon} \left \| \begin{pmatrix}
   \left (  \frac{d^*}{c^*} -\frac{b^*}{a^*} \right )D\\
   \left ( \frac{b}{a} - \frac{d}{c} \right )C
\end{pmatrix} \right 
\|_{\mathcal{H}}
 .
\end{equation}
 By \eqref{eq:fn_det_fixed_pt}, \eqref{eq:M_equal_Ainfty}, we have that $C$ and $D$ are almost everywhere bounded on $\T$
by
\[
 \left | C \left (\infty \right ) \right |^\frac{1}{2}  \leq (\cst \epsilon)^{-\frac{1}{2}}  ,
\]
where the last inequality follows from \eqref{eq:A_infty_estimate}.
Moreover, $a$ and $c$ are bounded below by $2^{-\frac 12}$ by assumption.
 Hence \eqref{firstx} gives
\[\|X\|_{\mathcal{H}}\le (\cst \epsilon)^{-\frac 32}
\left \| \begin{pmatrix}
     \frac{d^*}{c^*} -\frac{b^*}{a^*} \\
    \frac{b}{a} - \frac{d}{c} 
\end{pmatrix} \right 
\|_{\mathcal{H}}
\leq 2 (\cst \epsilon)^{-\frac 32}\left \| \begin{pmatrix}
     d^* a^* - b^* c^* \\
   b c - d a 
\end{pmatrix} \right \|_{\mathcal{H}}  .
\]

Adding and subtracting terms as in $bc-ad  = c(b-d) + d(c-a)$, and using
that $a,b,c,d$ are all bounded above by $1$ on $\T$, we obtain 
\begin{equation}\label{Xlipsch}\|X\|_{\mathcal{H}}\le 4 (\cst \epsilon)^{-\frac 32}
\left \| \begin{pmatrix}
     a \\
     b
\end{pmatrix} - \begin{pmatrix}
     c \\
     d
\end{pmatrix} \right \|_{\mathcal{H}}  . 
\end{equation}

By \eqref{abnormalized}, we have
\begin{equation*}
\left \| \begin{pmatrix}
    a_{+} \\ b_{+}
\end{pmatrix} - \begin{pmatrix}
    c_{+} \\ d_{+}
\end{pmatrix} \right \|_{\mathcal{H}} = \left \|  \frac{1}{A (\infty ) ^{\frac{1}{2}}} \begin{pmatrix}
    A \\
    B
\end{pmatrix} - \frac{1}{C (\infty ) ^{\frac{1}{2}}} \begin{pmatrix}
    C \\
D
\end{pmatrix} \right \|_{\mathcal{H}} 
\end{equation*}
\begin{equation}\label{a+lip}
\le \frac{1}{A (\infty ) ^{\frac{1}{2}}} \left \| X\right \|_{\mathcal{H} } +  \left \|\begin{pmatrix}
    C\\
    D
\end{pmatrix} \right \|_{\mathcal{H}} \left | \frac{1}{C (\infty ) ^{\frac{1}{2}}} - \frac{1}{ A (\infty )^{\frac{1}{2}}} \right |  . 
\end{equation}

Using \eqref{triplea} and
\eqref{definebe} the fact that $a_{-}  \left (\infty \right )\le 1$,
we have
\begin{align}\label{eq:A_bded_below}
\left | A (\infty ) ^{\frac{1}{2}} \right | = 
|a_{+} \left (\infty \right )|=\left | \frac{a  \left (\infty \right )}{a_{-}  \left (\infty \right )} \right | \geq 2^{-\frac 12} 
\end{align}
and analogously for $C (\infty )$. Further, as $(c,d) \in\mathbf{L}$, we have
\begin{equation}\label{secondlip}
\left \| \frac{1}{C (\infty )^\frac{1}{2}} \begin{pmatrix}
    C\\
    D
\end{pmatrix} \right \|_{\mathcal{H}} = \left \|  \begin{pmatrix}
    c\\
    d
\end{pmatrix} \right \|_{\mathcal{H}} =1  ,
\end{equation}
and using \eqref{eq:A_bded_below}, in particular 
\[
C (\infty ) ^\frac{1}{2} + A (\infty ) ^\frac{1}{2} \geq 2^{\frac{1}{2}}  ,
\]
we get
\[
   \left | {1} - \frac{C (\infty )^{\frac{1}{2}}}{A (\infty ) ^{\frac{1}{2}}} \right |=  \frac{ |C (\infty ) - A (\infty )|}{ A (\infty )^{\frac{1}{2}}  \left ( A (\infty ) ^\frac{1}{2} + C(\infty ) ^\frac{1}{2} \right )} 
\le 
\left | A (\infty ) -  C (\infty ) \right | 
\le  \int\limits_{\T} |A - C| 
\]
\begin{equation}\label{thirdlip}\leq 
\left \| X \right \|_{\mathcal{H}} .\end{equation}

Using \eqref{eq:norm_AB_estimate} and the above estimates
\eqref{eq:A_bded_below}, \eqref{secondlip}, and \eqref{thirdlip},
we obtain from \eqref{a+lip}
and \eqref{Xlipsch} 
\[
\left \| \begin{pmatrix}
    a_{+} \\
    b_{+}
    \end{pmatrix}-\begin{pmatrix}
    c_{+} \\
    d_{+}
    \end{pmatrix}\right \|_{\mathcal{H}} \leq 
    (2^{\frac 12}+1)
\left \|X \right \|_{\mathcal{H}}
    \le 
4(2^{\frac{1}{2}}+1) (\cst \epsilon)^{-\frac 32}
\left \| \begin{pmatrix}
     a \\
     b
\end{pmatrix} - \begin{pmatrix}
     c \\
     d
\end{pmatrix} \right \|_{\mathcal{H}}  .
\]
This proves Theorem \ref{thm:Lipschitz}.
\end{proof}

\begin{theorem} \label{layerlipschitz}
    The map sending $(a,b)$ to the coefficient $F_0$ of the sequence $F$ as in Theorem \ref{inverse} is Lipschitz 
    on $\mathbf{B}_{\epsilon}$ endowed with the metric $\mathcal{H}$.
    The Lipschitz constant is at most $(8 + 2^{ \frac 5 2}) (\cst \epsilon)^{-\frac{3}{2}}$, where $\cst$ is the constant in \eqref{abbound}.
\end{theorem}

\begin{proof}
Let $(a,b)$ and $(c,d)$ in $\mathbf{B}_{\epsilon}$ be the nonlinear Fourier series of $(F_n)$ and $(G_n)$. Let $(a_+, b_+)$ and $(c_+, d_+)$ be as in the proof of Theorem \ref{inverse}. We have 
\[2^{-\frac 12}\le {a_{+}  \left (\infty \right )}, {c_{+}  \left (\infty \right )}\le 1,\]
the upper bound for general elements in $\mathbf{L}$ and the lower bound
by \eqref{triplea} and assumptions on $a$ and $c$.
 By \eqref{eq:layer_stripping_explicit}, we have
\[
\frac{1}{2} \left | F_0 -G_0 \right | = \frac{1}{2}\left | \frac{b_+(0)}{a_+(\infty)}-\frac{d_+ (0)}{c_+ (\infty)} \right | \le \left | b_+ (0)c_+ (\infty)- d_+ (0)a_+ (\infty) \right |
\]

\[\le 
|(b_+ (0)- d_+ (0))c_+ (\infty)|+|d_+ (0)(c_+ (\infty)- a_+ (\infty))| \]
\[\le |b_+ (0)- d_+ (0)|+ |c_+ (\infty)- a_+ (\infty)| \le 
\int_{\T} |b_+ -d_+|+\int_{\T} |c_+- a_+ | 
\]
\[
\le 2^{\frac 12} \left \| \begin{pmatrix}
    a_+ \\
    b_+
\end{pmatrix} - \begin{pmatrix}
    c_+ \\
    d_+
\end{pmatrix} \right \|_{\mathcal{H}}  \leq (8+ 2^{ \frac 5 2}) (\cst \epsilon)^{-\frac{3}{2}} \left \| \begin{pmatrix}
    a \\
    b
\end{pmatrix} - \begin{pmatrix}
    c \\
    d
\end{pmatrix} \right \|_{\mathcal{H}}  . 
\]
Here the last inequality followed from Theorem \ref{thm:Lipschitz}.
This proves Theorem \ref{layerlipschitz}.
\end{proof}

\section{Proof of the main Theorem \ref{main}}

\label{mainproof}

Let $f$ be given as in Theorem \ref{main}.
Extend $f$ to an even function on $[-1,1]$.

We first show existence of the sequence $\Psi$ by construction.
Define
$b(z)=if(x)$ where $x=\cos \theta$ for $\theta\in [0,\pi]$ given by $z=e^{2i\theta}$. As $f$ is even, $b(z)=b(z^{-1})$ for  $z\in \T$, and in particular $b(z)$ is well-defined at $z=1$ because $f(-1) = f(1)$. Moreover, $b$ is purely imaginary and is bounded in absolute value by $2^{-\frac{1}{2}}-\epsilon$. By Theorem \ref{exista}, there is an $a$ such that $(a,b)\in \mathbf {B}_\epsilon$. By Theorem \ref{inverse}, there is a sequence
$F=(F_n) \in \ell^2 \left (\Z \right )$ so that $(a,b)$ is the nonlinear Fourier series of $F$.

The reflection symmetry of the purely imaginary $b$ implies $F_{-n} = F_n$ and $\bar{F_n} = F_n$ for all $n$. Indeed, by \eqref{eq:potentialreflection}, extended to infinite sequences,  and the fact that $b(z^{-1}) = b(z)$, the sequence $(F_{-n})$ has nonlinear Fourier series $(a ^* (z^{-1}), b(z))$. By the uniqueness part of Theorem \ref{exista}, this implies 
\[
a ^* (z^{-1}) = a(z)  .
\]
Thus $(a,b)$ is the nonlinear Fourier series of both $(F_n)$ and $(F_{-n})$, which by the uniqueness part of Theorem \ref{inverse} implies $F_n = F_{-n}$. And by \eqref{eq:potentialconjugate} extended to infinite sequences, the sequence $\bar{F_n}$ has nonlinear Fourier series 
\[
(a ^* (z^{-1}), b^* (z^{-1})) = (a (z), - b(z))  ,
\]
which by \eqref{eq:potentialmodulate}, again extended to infinite sequences, is the nonlinear Fourier series of $-F_n$. Again by uniqueness part of Theorem \ref{inverse}, we conclude that $\bar{F_n} = -F_n$, i.e., $F_n$ is purely imaginary.

We define $\psi_n\in (-\frac \pi 2,\frac \pi 2)$ for $n\in \N$ by
\[
\psi_n \equiv \arctan\left ( \frac{F_n}{i} \right )  .
\]

We now show the desired properties of the sequence $\Psi$. We begin with the Plancherel identity \eqref{plancherel}. We compute
\[
   \frac{1}{\pi} \int\limits_{-1} ^1 \log \left ( 1 - f(x)^2 \right ) \frac{dx}{\sqrt{1-x^2} }= \frac{1}{\pi}\int_0^\pi \log(1-f(\cos \varphi)^2 ) \, d\varphi \]
   \[=\int\limits_{\T} \log \left ( 1 - \left | b \right |^2 \right )= 2 \int\limits_{\T} \log \left | a \right |  = 2 \log \left | a^* \left ( 0 \right ) \right |  , 
\]
where the last equality follows from $a$ being outer. By 
\eqref{fullplancherel}, the last term equals
\[
 - \sum\limits_{k \in \Z} \log (1+|F_k|^2)  = -\sum\limits_{k \in \Z} \log (1+\tan ^2 \psi_{|k|} )  .
\]
This proves \eqref{plancherel}.

Next we show convergence of $\Im (u_d \left (\Psi, x \right )$ to $f$ in the norm \eqref{hsnorm}.
Let $(a_d, b_d)$ be the nonlinear Fourier series of the truncated sequence $(F_n 1_{\{|n| \leq d \}})$. Then by Lemma \ref{upperleft}, we have 
\[
b_d \left (z \right ) = \Im \left (u_d \left (\Psi, x \right ) \right )  ,
\]
where we again recall that the right side is an even function on $[-1,1]$. 
By the reasoning just above \eqref{fullplancherel}, the
sequence $(a_d, b_d)$ converges to $(a,b)$ in $\mathbf{L}$, and hence $b$ converges to $b_d$ in $L^2 \left (\T \right )$. As
\begin{align}\label{eq:change_of_vars} 
\left ( \frac{1}{\pi} \int\limits_{-1}^1 \left | \Im  u_d \left (\Psi, x\right ) - f \left (x \right ) \right |^2 \frac{dx}{\sqrt{1-x^2}} \right ) ^{\frac{1}{2}}  = \left \| b_d - b \right \|_{L^2 \left ( \T \right )}  ,
\end{align}
which converges to $0$ as $d \to \infty$ by the remarks just above \eqref{fullplancherel}, we then have convergence of $\Im u_d \left (\Psi, x\right )$ to $f$ in the norm \eqref{hsnorm}. 

This shows existence of $\Psi$. 
To see uniqueness, let $\tilde{\Psi}$ be any sequence satisfying the properties of the theorem. Set $\tilde{F}_n =  i \tan \tilde{\psi}_{|n|}$. By \eqref{plancherel}, the sequence $\tilde{F}$ is square summable. Let $(\tilde{a},\tilde{b})$ be its non-linear Fourier series, and $(\tilde{a}_d,\tilde{b}_d)$ the non-linear Fourier series of its truncations. By the remarks just above \eqref{fullplancherel}, $\tilde{b}_d$ converges to $\tilde{b}$, and it also converges to $b$ by definition of $b$ and the convergence assumption of the theorem.
Hence $b=\tilde{b}$. Because 
\[
\left | \frac{\tilde{a^*}}{a^*} \right | = 1
\] on $\T$ and $\frac{1}{a^*} \in H^2 (\D)$, then $\frac{\tilde{a^*}}{a^*}$ is an inner function on $\D$. By \eqref{fullplancherel}, \eqref{plancherel}, the definition of the function $b = i f$ and then the definition \eqref{eq:b_to_a} of the outer function $a$, we have
\[
\tilde{a} (\infty) = \prod\limits_{n \in \Z} (1 + |\tilde F_{n}|^2)^{ - \frac 1 2} = \frac{1}{2} \int\limits_{\T} \log (1-|\tilde b|^2) =\int\limits_{\T} \log (|a|) = a (\infty)\, .
\]
Thus $\frac{\tilde{a} (\infty)}{a (\infty)} =1$ and by the maximum principle, the inner function $\frac{\tilde{a} ^*}{a^*}$ must be constant $1$, i.e., $a=\tilde{a}$. By the uniqueness part of Theorem \ref{inverse}, we have $\tilde{F}_n=F_n$. 
Hence $\tilde{\Psi}=\Psi$ since for each $j$ we have $\psi_j, \tilde \psi_j \in (- \frac{\pi}{2}, \frac{\pi}{2})$, an interval on which $\tan$ is injective.

Now we show that the map sending $f$ to $\Psi$ is Lipschitz.
It suffices to show for each $k\ge 0$ that the map from $f$ to
$\psi_k$ is Lipschitz.

We write this map as a composition of three maps.
By Theorem \ref{layerlipschitz}, the map  sending 
 $(a,b)$ in $\mathbf{B}_{\epsilon}$
 to $F_0$ is Lipschitz. As the shift $(a,b) \mapsto (a,b z)$ is an isometry in $\mathbf{B}_{\epsilon}$, the same holds for $F_k$ with $k\in \Z$ and we have for $(a,b)$ and $(\tilde{a},\tilde{b})$ in $\mathbf{B}_{\epsilon}$,
\begin{align}\label{eq:ab_to_F}
\left | F_k -\tilde{F}_{k}  \right | \leq (8+ 2^{ \frac 5 2}) (\cst \epsilon)^{-\frac{3}{2}} \left \| \begin{pmatrix}
    a\\
    b
\end{pmatrix} - \begin{pmatrix}
    \tilde{a}\\
    \tilde{b}
\end{pmatrix}
\right \|_{\mathcal{H}} , 
\end{align}
where $\cst$ is as in \eqref{abbound}.

As $\arctan \left (x \right )$ has slope between $-1$ and $1$, we have
\begin{align}\label{eq:Lip_F_to_psi}
\left | \psi_k -\tilde{\psi}_{k}  \right | \leq \left | F_k -\tilde{F}_{k}  \right |   . 
\end{align}

It remains to obtain a Lipschitz bound for the map sending $f$ to $(a,b)$.
 By the cosine theorem, with  $\theta(z)$
 the angle between $a(z)$ and $\tilde{a}(z)$
\begin{equation*}
|a- \tilde{a}|^2=|a|^2+|\tilde{a}|^2-2|a||\tilde{a}|\cos \theta 
=(|a|-|\tilde{a}|)^2 + 2 |a||\tilde{a}|(1-\cos \theta ) 
\end{equation*} 
\begin{equation}
\label{costheorem1}
    \le (|a|-|\tilde{a}|)^2 + 2 (1-\cos \theta ) 
    \le (|a|-|\tilde{a}|)^2 + \theta^2 .
\end{equation} 
Here we used that  $|a|,|\tilde{a}|\le 1$, and
that $2(1-\cos \theta )$ vanishes  of order two at $\theta=0$ and has second derivative less than 
or equal to two.

The angle $\theta$ is given by the imaginary part of $\log(a)-\log(\tilde{a})$.
As $\log(a)$ and $\log(\tilde{a})$ have  analytic extensions to $\D^*$ that are real at $\infty$, the angle $\theta$ 
is dominated in absolute value by the imaginary part of $\log(a)-\log(\tilde{a})$, which in turn is given as $-H(\log|a|-\log|\tilde{a}|)$ for the Hilbert transform $H$. Recall that the Hilbert transform 
has operator norm bounded by $1$ on $L^2(\T)$.

Inequality \eqref{costheorem1} yields 
\begin{equation}\label{asquared}
  \left \| a - \tilde{a} \right \|_{L^2 \left (\T \right )}^2 \leq 
  \left \|  \left | a \right|   -   \left | \tilde{a} \right|   \right \|_{L^2 \left (\T \right )}^2  + \left \|  H({\log   \left | a \right|}  -  {\log  \left | \tilde{a} \right|})   \right \|_{L^2 \left (\T \right )}^2 ,
\end{equation}
\begin{equation}\label{asquared1}
 \leq 
  \left \|   \sqrt{1- |b|^2}   -    \sqrt{1- |\tilde{b}|^2}   \right \|_{L^2 \left (\T \right )}^2  + \frac 14\left \|  {\log   \left | 1- |b|^2 \right|}  -  {\log   | 1- |\tilde{b}|^2 |}   \right \|_{L^2 \left (\T \right )}^2 , \nonumber
\end{equation}
Using that on the interval  $[0, 2^{-\frac 12}]$, the map 
$x \mapsto \log (1-x^2)$ has slope bounded by $2^{\frac 32}$  and $x \mapsto \sqrt{ 1-x^2}$ has slope bounded by $2^{\frac 12}$, we estimate the 
\eqref{asquared} by

\[\le 
4\left \| \left | b \right| -  \left | \tilde{b} \right| \right \|_{L^2 \left (\T \right )} ^2 \leq 4 \left \| b -  \tilde{b} \right \|_{L^2 \left (\T \right )} ^2.
\]
We obtain
\begin{align}\label{eq:Lip_f_to_ab}
  \left \| \begin{pmatrix}
    a\\
    b
\end{pmatrix} - \begin{pmatrix}
    \tilde{a}\\
    \tilde{b}
\end{pmatrix}
\right \|_{\mathcal{H}} ^2 \leq 5\left \| b - \tilde{b} \right \|_{L^2 \left ( \T \right )} ^2 = \frac{5}{\pi} \int\limits_{-1}^1 \left | f \left (x \right ) - \tilde{f} \left (x \right )\right |^2 \frac{dx}{\sqrt{1-x^2}}   .
\end{align}
Combining  \eqref{eq:ab_to_F}, \eqref{eq:Lip_F_to_psi} and \eqref{eq:Lip_f_to_ab}, we obtain
\[
\left |\psi_0 -\tilde{\psi}_0  \right | \leq 5 ^{ \frac 1 2} (8+ 2^{ \frac 5 2}) (\cst \epsilon)^{-\frac{3}{2}}\left ( \frac{1}{\pi} \int\limits_{-1}^1 \left | f \left (x \right ) - \tilde{f} \left (x \right )\right |^2 \frac{dx}{\sqrt{1-x^2}} \right ) ^{\frac{1}{2}}  . 
\]
The bound for $\cst$ from Lemma \ref{lemma:best_cst_eta} yields the Lipschitz constant above is at most $7.3 \epsilon^{- \frac 3 2}$, which completes the proof of Theorem \ref{main}.

\section*{Funding and Conflicts of interests}

The authors acknowledge support by the Deutsche Forschungsgemeinschaft (DFG, German Research Foundation) under Germany's Excellence Strategy -- EXC-2047/1 -- 390685813 as well as CRC 1060.

The authors have no relevant financial and non-financial interests to disclose.

\bibliographystyle{amsalpha}

\bibliography{references}

\end{document}